\documentclass[11pt,reqno]{amsart}
\RequirePackage[margin=1.2in]{geometry}
\RequirePackage{multirow}
\usepackage{pkgfile}
\usepackage{mathtools,subcaption}
\usepackage{diagbox}
\usepackage{multirow,nicefrac}
\usepackage{rotating}
\usepackage{mathrsfs}
\usepackage{xifthen,setspace}
\usepackage{bbm}
\usepackage{url} 
\usepackage{float}
\newtheorem{example}{Example}

\newcommand{\sg}{{\sG}}

\newcommand{\sx}{{\mathscr{X}}}
\allowdisplaybreaks

\begin{document}
\title[Multisample Tests Based on Optimal Matchings]{Distribution-Free Multisample Test Based on Optimal Matching with Applications to Single Cell Genomics}

\author[]{Divyansh Agarwal\textsuperscript{$\dagger$}}\thanks{\textsuperscript{$\dagger$}The first two authors contributed equally to the paper.}
\address{Graduate Group in Genomics and Computational Biology, Medical Scientist Training Program, Perelman School of Medicine, University of Pennsylvania} 
\email{Divyansh.Agarwal@pennmedicine.upenn.edu}

\author[Agarwal, Mukherjee, Bhattacharya, Zhang]{Somabha Mukherjee\textsuperscript{$\dagger$} \and Bhaswar B.\ Bhattacharya \and Nancy R. \ Zhang}
\address{Department of Statistics, The Wharton School, University of Pennsylvania}
\email{\{somabha,bhaswar,nzh\}@wharton.upenn.edu}

\date{\today}


\spacing{1.125}
\maketitle

\begin{abstract}
In this paper we propose a nonparametric graphical test based on optimal matching, for assessing the equality of multiple unknown multivariate probability distributions. Our procedure pools the data from the different classes to create a graph based on the minimum non-bipartite matching, and then utilizes the number of edges connecting data points from different classes to examine the closeness between the distributions. The proposed test is exactly distribution-free (the null distribution does not depend on the distribution of the data) and can be efficiently applied to multivariate as well as non-Euclidean data, whenever the inter-point distances are well-defined. We show that the test is universally consistent, and prove a distributional limit theorem for the test statistic under general alternatives. Through simulation studies, we demonstrate its superior performance against other common and well-known multisample tests.  In scenarios where our test suggests distributional differences across classes, we also propose an approach for identifying which class or group contributes to this overall difference. The method is applied to single cell transcriptomics data obtained from the peripheral blood, cancer tissue, and tumor-adjacent normal tissue of human subjects with hepatocellular carcinoma and non-small-cell lung cancer. Our method unveils patterns in how biochemical metabolic pathways are altered across immune cells in a cancer setting, depending on the tissue location. All of the methods described herein are implemented in the \textsf{R} package \textsf{multicross}.
\end{abstract}

\section{Introduction}
\label{sec:intro}

Given $K$ multivariate probability distributions $F_1,F_2,\ldots, F_K$, the {\it $K$-sample problem} is to test the hypotheses
\begin{align}\label{eq:Ksample}
H_0: F_1 = \cdots = F_K \quad \textrm{versus} \quad  H_1: F_s \neq F_t, \quad \textrm{for some} ~1\leq s < t \leq K.
\end{align}
This is a classical problem in statistical inference which has been extensively studied in the parametric regime, where the distributions are assumed to have certain, low-dimensional functional forms. Parametric methods, however, often perform poorly for misspecified models and for high-dimensional problems, especially when the number of nuisance parameters is large. This necessitates the development of non-parametric methods, which make no distributional assumptions on the data, but are still powerful for a wide class of alternatives. Moreover, with the recent accumulation of high-dimensional and non-Euclidean data arising from genomics, social networks, bioinformatics, and finance, it is imperative to develop non-parametric methods which are computationally efficient, robust and applicable to the various kinds of modern data types. In this paper, we consider non-parametric tests for the $K$-sample problem, which are  exactly {\it distribution-free}, that is, tests for which the null distribution does not depend on the underlying (unknown) distribution of the data. This property is particularly desirable, because such tests can be directly calibrated under the null irrespective of the distribution or type of the data, making them readily applicable in a wide range of problems. 

Nonparametric testing of two multivariate distributions has a long history, which has spawned renewed interest in light of modern applications. For univariate data,  there are several celebrated distribution-free two-sample tests such as the Kolmogorov-Smirnov maximum deviation test \cite{smirnov}, the Wald-Wolfowitz runs test \cite{ww}, and the Mann-Whitney rank-sum test \cite{mann_whitney} (see the textbook \cite{gc} for more on these tests). Efforts to generalize these methods to higher dimensions go back to Weiss \cite{weiss} and Bickel \cite{bickel}. Friedman and Rafsky \cite{fr} proposed the first computationally efficient 2-sample test, which applies to high-dimensional data. The Friedman-Rafsky test, which can be viewed as a generalization of the univariate runs test,  computes the Euclidean minimal spanning tree (MST)\footnote{Given a finite set $S \subset \R^d$, the {\it minimum spanning tree} (MST) of $S$ is a connected graph  with vertex-set $S$ and no cycles, which has the minimum weight, where the weight of a graph is the sum of the distances of its edges.}  of the pooled sample, and rejects the null if the number of edges with endpoints in different samples is small.  Variants of this test based on nearest-neighbor graphs were considered by Henze \cite{hn} and Schilling \cite{sch}. Recently, Chen and Friedman \cite{chfr} suggested novel modifications of this method for high-dimensional and object data, and Chen et al. \cite{hcII} proposed new and powerful tests to deal with the issue of unequal sample sizes. Asymptotic properties of these tests can be studied in a general asymptotic framework introduced in \cite{bh}. Another computationally efficient 2-sample test based on the concept of multivariate ranks, defined using optimal transport, was recently proposed by Ghosal and Sen \cite{ghoshal_sen}.

Even though many of these methods can be effectively used in high-dimensional problems, none of them inherit the exact distribution-free property of the univariate tests.  A breakthrough in this direction was made relatively recently by Rosenbaum \cite{rosen}, who proposed the {\it crossmatch test}, a multivariate two-sample test based on the minimum non-bipartite matching (Definition \ref{mdm}) of the pooled sample. This test is exactly distribution-free in finite samples and computationally efficient (the test statistic can be computed in time which is polynomial in both the number of samples and the dimension of the data), making it particular attractive for high-dimensional applications. This test has also found  many interesting applications in causal inference, especially in assessing balance between covariates in a treatment group and a matched control group \cite{hcds,crossmatch_sensitivity,covariate_balance,mbs}. More recently, Biswas et al. \cite{biswas} proposed another two-sample test based on Hamiltonian cycles, which is also distribution-free in finite samples. However, unlike the  minimum non-bipartite matching, computing the minimum weight Hamiltonian path is NP-hard, making this test computationally prohibitive beyond small sample sizes. 

Here, we study nonparametric distribution-free tests for the general $K$-sample problem \eqref{eq:Ksample}. As expected, this problem is well-understood in dimension 1.  Mood \cite{mood} considered the $K$-sample generalization of the runs test, and Kruskal and Wallis \cite{ks,kswallis} derived the $K$-sample analogue of the Mann-Whitney test, both of which are distribution-free. Our interest is in devising efficient distribution-free methods, which are powerful for a wide range of alternatives in arbitrary dimensions. We are motivated by applications  in high-throughput biological experiments, where the multisample problem often arises. For instance, it is not uncommon to examine the distribution of a set of high-dimensional features across various models or conditions. Recently, single cell technologies have made it possible to profile the expression of tens of thousands of genes across thousands of cells. The cells might belong to different subtypes, where  the $K$ types are characterized based on some functional or morphological  parameter. In this setting, determining whether the expression of a set of genes, corresponding to a particular biochemical pathway or function, belong to the same or different distribution across the $K$ groups can yield insights into cellular processes and the underlying biological system.


Even though the high-dimensional multisample problem manifests itself in various modern applications, methodological progress to address this problem has been limited.  A nearest-neighbor based test for testing the equality of multiple distributions with categorical components was considered in \cite{nettleton_banerjee}.  Recently, Petrie \cite{AdamPet} considered the direct generalization of the Friedman-Rafsky and crossmatch tests, which counts the number of edges across the different samples in the geometric graph (MST or matching) constructed on the pooled sample. However, this test tends to lose power with increases in dimension and/or the number of classes, and the mathematical properties of this test have not been investigated. 


We propose a new graph-based multisample test based on optimal matchings. To compute the test statistic, we construct the minimum non-bipartite matching of the pooled sample, compute the $K\times K$ matrix of cross-counts (the $(s, t)$-th element of this matrix is the  number of edges in the matching from sample $s$ to sample $t$), then combine these counts using their Mahalanobis distance. We show that this test is exactly distribution-free under the null (Proposition \ref{ppn:H0_distribution}),  derive its asymptotic null distribution (Theorem \ref{nullasmS}), and demonstrate its consistency under general alternatives (Theorem \ref{consmain}). We also prove a conditional central limit theorem (CLT) of the entire vector of cross-counts under the alternative (Section \ref{condasm1}).  More precisely, we show that the cross-counts, centered by their means, conditional on the pooled sample and scaled appropriately, converge in distribution to a multivariate normal distribution under general alternatives. As a consequence, we obtain a distributional limit theorem for the proposed test statistic under the alternative, which to the best of our knowledge, is a new result even for the 2-sample case (where the proposed test statistic is equivalent to Rosenbaum's 2-sample crossmatch test). Therefore, this result adds to our theoretical understanding of the crossmatch and general matching-based tests, which includes the proposed method and the baseline generalizations considered in \cite{AdamPet}. In Section \ref{simulsec} we compare the power of our test with other existing tests for various alternatives. Our experiments demonstrate that the proposed method outperforms other relevant parametric and non-parametric tests in a diverse set of simulation settings.

Lastly, we demonstrate the broad potential of our method on real data examples from the single cell biology domain, where we use our test to investigate whether the activities of biological pathways across $K$ closely-related cell types are conserved. Single cell data is sparse count data with many zeros, so it is difficult to transform it to conform to normal distribution assumptions. Moreover, gene expression has complex correlation structure and thus any parametric model would require many nuisance parameters. Therefore, our novel crossmatch-based method, being nonparametric and distribution free, is especially fitting. We show the utility of our proposed test in the comparison of the distribution of gene sets (as a proxy for examining biochemical pathways) across cell populations using single cell RNA-sequencing (scRNA-seq) data (Section \ref{sec:applications}). Our method successfully recapitulates known biology by detecting differential distribution of metabolic pathways which are known to be disparate across T cell subtypes. We also discovered pathways such as purine metabolism that differed across T cell subtypes, irrespective of whether the sequenced cells were obtained from blood, liver cancer tissue, lung cancer tissue, or tumor-adjacent normal tissue. Moreover, we ascertain that our method can be used to narrow down gene sets that are differentially distributed across cell types, thereby facilitating its use in single cell clustering and visualization algorithms. Altogether, as illustrated through real case studies, our proposed method combines statistical innovation to answer practical, scientifically relevant questions. 


\section{Multisample distribution-free tests based on optimal matching}\label{sec1}

Recall the $K$-sample hypothesis \eqref{eq:Ksample}, and assume that for every $s \in [K]:=\{1, 2, \ldots, K\}$, we are given $N_s$ i.i.d. observations $\bm X^{(s)} := \{X_1^{(s)},X_2^{(s)},\ldots,X_{N_s}^{(s)} \}$ from the distribution $F_s$.  In this section we describe a novel distribution-free, computationally efficient $K$-sample test, based on  the minimum non-bipartite matching, which can be readily used for data in arbitrary metric spaces, such as high dimensional data, functional data, and object data. 

We begin with the formal definition of a minimum non-bipartite matching. For simplicity, we will assume throughout that the total number of samples $N=\sum_{s=1}^K N_s:=2I$ is even; otherwise, we can add or delete a sample point to make it even. 

\begin{defn}\label{mdm}
Given a finite $S\subset \R^d$ and a symmetric distance matrix $D:=((d(a, b)))_{a\ne b \in S}$, a {\it non-bipartite matching} of $S$ is a partition of the elements of $S$ into $I=\frac{N}{2}$ non-overlapping sets of size 2 each, that is, $$S=S_1 \bigcup S_2 \bigcup \cdots \bigcup S_{I}, \quad \text{where } |S_a|=2 \text{ and } S_a\cap S_b=\emptyset,  \text{ for } 1 \leq a \ne b \leq I.$$ The {\it weight} of a non-bipartite matching is the sum of the distances between the $I$ matched pairs. A {\it minimum non-bipartite matching} of $S$ is a matching which has the minimum weight over all matchings of $S$. (In case of multiple minimizers any one of them can be chosen.) The {\it minimum non-bipartite matching graph} $\sG(S)=(V(\sG(S)), E(\sG(S)))$ is the graph with vertex set $V(\sG(S))=S$ and edge set $E(\sG(S))=\{S_1, S_2, \ldots, S_{I} \}$, consisting of the $I$ disjoint pairs. 
\end{defn}


The 2-sample {\it cross-match} (CM) test proposed by Rosenbaum \cite{rosen} rejects the null hypothesis in \eqref{eq:Ksample} for small values of
\begin{align}\label{eq:R2N}
R_{2, N}:=\sum_{i=1}^{N_1}\sum_{j=1}^{N_2}\bm 1 \{(X_i^{(1)}, X_j^{(2)} ) \in E(\sG(\mathscr{X})) \},
\end{align}
where $\mathscr{X} := \{\bm X^{(1)},\bm X^{(2)} \}$ is the pooled sample. Note that the statistic $R_{2, N}$ counts the number of matched edges in the pooled sample  with one end point in  sample 1 and the other endpoint in sample 2 (the ``cross-matches''), which is expected to be small when the two distributions are  different. 

\begin{figure}[h]
\begin{center}
\includegraphics[width=4.5in]{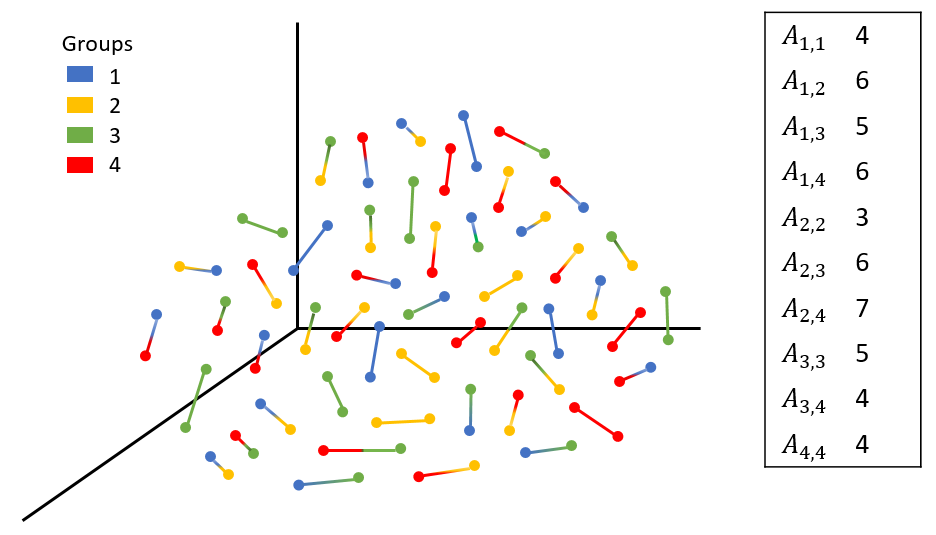}
\end{center}
\caption{\small{Illustration of a minimum non-bipartite matching for $100$ points in 3-dimensions with $4$ classes, and the different cross/pure counts.}} 
\label{fig:illustratet}
\vspace{-0.1in}
\end{figure}

Here, we consider two generalizations of the CM statistic when there are more than 2 samples, based on the minimum non-bipartite matching of the pooled sample. In this case, denoting the pooled sample by $\mathscr{X} := \{\bm X^{(1)},\bm X^{(2)},\ldots,\bm X^{(K)} \}$, we define, for $1 \leq s \ne t \leq K$,  the $(s, t)$-\textit{cross count} as the number of matched edges in the pooled sample with one endpoint in sample $s$ and the other end-point in sample $t$, which is denoted by 
\begin{align}\label{eq:st_count}
a_{st}(\sg(\sx)) :=   \sum_{i=1}^{N_s}\sum_{j=1}^{N_t}\bm 1 \{(X_i^{(s)}, X_j^{(t)} ) \in E(\sG(\mathscr{X})) \}. 
\end{align}
For each $s \in [K]$, we also define the $(s, s)$-\textit{pure count} as the number of matched edges in the pooled sample with both endpoints in sample $s$, which is denoted by 
\begin{align}\label{eq:ss_count}
a_{ss}(\sg(\sx)) :=  \frac{1}{2} \sum_{i=1}^{N_s}\sum_{j=1}^{N_s}\bm 1 \{(X_i^{(s)}, X_j^{(s)} ) \in E(\sG(\mathscr{X})) \}. 
\end{align}
The matrix of cross/pure counts ${\bm A}_N(\sg(\sx)) = \left(a_{st}(\sg(\sx))\right)_{1\leq s,t\leq K}$ will be referred to as the {\it count matrix}. (Hereafter, we will drop the dependence on $\sG(\sX)$ from $\bm A_N(\sG(\sX))$ and its elements $a_{st}(\sg(\sx)$, whenever it is clear from the context.) Figure \ref{fig:illustratet} shows the cross/pure counts for a sample of 100 points in 3-dimensions with 4 classes.


We show below in Proposition \ref{ppn:H0_distribution}, that the joint distribution of the elements of the count matrix is exactly distribution-free under the null. Therefore, we can construct distribution-free tests for \eqref{eq:Ksample} by considering real-valued functions of the cross-matrix, as follows:

\begin{itemize}

\item  The \textit{multisample crossmatch} (MCM)  test rejects the null in  \eqref{eq:Ksample} for small values of  the statistic 
\begin{align}\label{eq:RKN}
R_{K, N}:= \sum_{1\leq s<t\leq K} a_{st}(\sG(\sX)). 
\end{align}
This direct generalization of the 2-sample cross-match statistic (recall \eqref{eq:R2N}), which counts the total number of cross edges in the matching constructed using the pooled sample, was considered in \cite{AdamPet}. Here, we derive its asymptotic properties, and use it as a baseline for empirical comparisons. 

\item There are many natural multivariate alternatives where the MCM test described above performs poorly, especially when the dimension is large and the number of groups is big. 
To circumvent this issue, we propose a new test statistic based on the Mahalanobis distance of the observed cross-counts, which rejects the null for large values of 
\begin{align}\label{eq:SKN}
S_{K, N}:= \left(\vec{\bm A}_N - \mathbb{E}_{H_0} \vec{\bm A}_N\right)^\top \Cov_{H_0}^{-1}(\vec{\bm A}_N)\left(\vec{\bm A}_N - \mathbb{E}_{H_0} \vec{\bm A}_N\right),
\end{align}
where $\vec{\bm A}_N$ is the vector of length ${K \choose 2}$ corresponding to the cross-counts (the upper-triangular part of $\bm A_N$),\footnote{More formally, $\vec{\bm A}_N:= \left(a_{12}, \ldots, a_{1K}, a_{23}, \ldots, a_{2K}, \ldots, a_{K-2,K-1}, a_{K-2,K}, a_{K-1,K}\right)^\top$, the vector obtained by concatenating the rows of $\bm A_N$ in the upper triangular part.} and $\mathbb{E}_{H_0}(\vec{\bm A}_N)$ and $\mathrm{Cov}_{H_0}(\vec{\bm A}_N)$ denote the mean and the covariance matrix of $\vec{\bm A}_N$ under the null hypothesis, respectively (exact formulas are given below in Proposition \ref{ppn:mean_var} and the invertibility of $\mathrm{Cov}_{H_0}(\vec{\bm A}_N)$ is proved in Lemma \ref{invert1}). We refer to this test as the \textit{multisample Mahalanobis crossmatch} (MMCM) test.\footnote{When $K=2$ (the two-sample problem), the tests based on $R_{2, N}$ and $S_{2, N}$ are equivalent: In this case, the vector $\vec{\bm A}_N$ has only 1 element which is the number of cross-matches $a_{11}$, and \eqref{eq:SKN} simplifies to $S_{2, N}=\frac{(R_{2, N}-\E_{H_0}(R_{2, N}))^2}{\Var_{H_0}(R_{2, N})}$,
which is the square of the standardized CT statistic \eqref{eq:R2N}.} Note that adjusting by the sample covariance matrix brings the cross-counts in the same scale, which makes $S_{K, N}$ a more appropriate measure of the centrality of the empirical cross-counts, leading to significant power improvements when $K$ becomes large.

\end{itemize}

\subsection{Exact Null Distribution}

The following proposition shows that the  joint distribution of the elements of the count matrix is distribution-free under the null, that is, it does not depend on the unknown distribution $F_1=\cdots=F_K$.

\begin{ppn}\label{ppn:H0_distribution} Let $\bm A_N=((a_{st}))_{1\leq s, t \leq K}$ be as defined in \eqref{eq:st_count} and  \eqref{eq:ss_count}. Then 
\begin{equation}\label{nulld} 
\mathbb{P}_{H_0}\left({\bm A}_N = \bm b \Big| \sx\right) = \frac{1}{\binom{N}{N_1,\ldots,N_K}} \cdot \frac{2^{\sum_{1\leq s<t\leq K}b_{st}} I!}{\prod_{1\leq s \leq t\leq K} b_{st}!}, \quad \text{for} \quad \bm b =((b_{st}))_{1\leq s, t \leq K} \in\sB,
\end{equation}
where $\sB$ is the set of all symmetric $K\times K$ matrices $\bm{b} = ((b_{st}))$ with non-negative integer entries, satisfying $2 b_{ss} + \sum_{t \neq s} b_{st} = N_s$, for all $s \in [K]$. As a consequence, the statistics $R_{K, N}$ and $S_{K, N}$ defined above, are distribution-free under $H_0$. 
\end{ppn}

\begin{proof} Note that 
$$2 a_{ss} + \sum_{t \neq s} a_{st} = N_s, \quad \text{for each } s \in [K],$$
since all edges in the graph $\sG(\sX)$ are disjoint,  and each $(s,s)$-edge has both endpoints in sample $s$ and each $(s, t)$-edge, where $s \ne t$, has one of its endpoints in sample $s$. Therefore, the distribution of the cross-count matrix $\bm A_N$ is supported on the set $\sB$ defined above. 

Now, given $\bm{b} \in \sB$ and the pooled sample $\sx$, there are $$\frac{2^{\sum_{1\leq s<t\leq K}b_{st}} I!}{\prod_{1\leq s \leq t\leq K} b_{st}!}$$ ways of forming the classes $\bm X^{(1)},\ldots,\bm X^{(K)}$, such that $a_{st} = b_{st}$ for all $s,t \in [K]$. This comes from  first assigning the $I$ matched pairs such that there are $b_{st}$ pairs corresponding to the  $(s, t)$-counts, for $1 \leq s \leq t \leq K$, in $\frac{I!}{\prod_{1\leq s \leq t\leq K} b_{st}!}$ ways, and then, for each of the $b_{st}$ cross-matched pairs, assigning either one of the end points $s$ or $t$ in $2^{b_{st}}$ ways, for $1 \leq s < t \leq K$. Since, the random vector $(X_1^{(1)},\ldots, X_{N_1}^{(1)},\ldots, X_1^{(K)},\ldots, X_{N_K}^{(K)})$ is exchangeable under $H_0$, each of these classifications has probability $\binom{N}{N_1,\ldots,N_K}^{-1}$. Hence, \eqref{nulld} follows.  

Note that the RHS of \eqref{nulld} does not depend on $\sx$, which implies $\mathbb{P}_{H_0}({\bm A}_N = \bm b | \sx)= \mathbb{P}_{H_0}({\bm A}_N = \bm b)$. Therefore,  the statistics $R_{K, N}$ and $S_{K, N}$, which are functions of the matrix $\bm A_N$, are distribution-free under $H_0$.  
\end{proof}

The mean and covariances of this distribution, which are required for computing the MMCM statistic, can be easily derived:

\begin{ppn}\label{ppn:mean_var}
Let $\vec{\bm A}_N$ be as in \eqref{eq:SKN}. The entries of the mean vector of $\mathbb{E}_{H_0} \vec{\bm A}_N$ are given by: 
\begin{align}\label{eq:expectationH0}
\mathbb{E}_{H_0} (a_{st}) =   \left\{
\begin{array}{lll}
\frac{N_s N_t}{N-1}& \textrm{if} & s < t, \\
\frac{N_s(N_s-1)}{2(N-1)} & \textrm{if}& s= t.
\end{array} 
\right.
\end{align}
The entries of the covariance matrix $\mathrm{Cov}_{H_0}(\vec{\bm A}_N)$ are as follows: \begin{itemize}
\item  If $1\leq s_1 \ne s_2 \leq K$, $\mathrm{Var}_{H_0}(a_{s_1 s_2})   = \frac{N_{s_1} N_{s_2}(N_{s_1}-1)(N_{s_2}-1)}{(N-1)(N-3)} + \frac{N_{s_1} N_{s_2}}{N-1}\left(1-\frac{N_{s_1}N_{s_2}}{N-1}\right).$ \\ 

\item If $1\leq s_1 \ne s_2 \ne s_3 \leq K$, $\mathrm{Cov}_{H_0}(a_{s_1 s_2},a_{s_1 s_3})    =   \frac{N_{s_1}(N_{s_1}-1)N_{s_2} N_{s_3}}{(N-1)(N-3)} - \frac{N_{s_1}^2 N_{s_2} N_{s_3}}{(N-1)^2}$. \\ 

\item If $1\leq s_1 \ne s_2 \ne s_3 \ne s_4 \leq K$, $\mathrm{Cov}_{H_0}(a_{s_1 s_2}, a_{s_3 s_4}) = \frac{2N_{s_1} N_{s_2} N_{s_3} N_{s_4}}{(N-1)^2(N-3)}$. 
\end{itemize}
\end{ppn}

The proof of the proposition is given in Appendix \ref{sec:pfmean_var}. It follows by a direct combinatorial analysis and observing that, under the permutation null distribution, all possible ${N \choose {N_1, N_2, \ldots, N_K}}$ relabelings of the data are equally likely.

\subsection{Asymptotic Null Distribution}

In theory, the exact cutoff for the MMCM test can be obtained using the quantiles of the distribution in \eqref{nulld}. Another alternative is to use the exchangeability of the data under $H_0$, and perform a permutation test. However, both these approaches are computationally cumbersome when the sample size increases. In this case, it is more convenient to use rejection regions based on the asymptotic null distribution. This is derived in the theorem in the usual limiting regime where $N \rightarrow \infty$ such that 
\begin{align}\label{eq:Nlimit}
\left(\frac{N_1}{N}, \frac{N_2}{N}, \ldots, \frac{N_K}{N}\right) \rightarrow (p_1, p_2, \ldots,p_K) \in (0,1)^K,
\end{align}
where $\sum_{s=1}^K p_s=1$. 

\begin{thm}\label{nullasmS} Under the null $H_0$,
\begin{equation}\label{ant}
\Cov_{H_0}^{-\frac{1}{2}}(\vec{\bm A}_N)\left(\vec{\bm A}_N - \mathbb{E}_{H_0}\vec{\bm A}_N\right) \xrightarrow{D}  N_{\binom{K}{2}}(0, \mathrm I). 
\end{equation} 
This implies, under $H_0$, the MMCM statistic $S_{K, N} \xrightarrow{D} \chi_{\binom{K}{2}}^2$, as $N \rightarrow \infty$, and the test with rejection region 
\begin{align}\label{eq:SKNregion}
\left\{S_{K, N} >  \chi_{\binom{K}{2}, 1-\alpha}^2\right\},
\end{align}
is asymptotically level $\alpha$.\footnote{For $p \geq 1$, $N_p(\bm \mu, \Sigma)$ denotes the multivariate normal distribution with mean $\bm \mu \in \R^p$ and covariance matrix $\Sigma \in \R^{p \times p}$. Moreover, $\chi^2_n$ denotes the chi-squared distribution with $n$ degrees of freedom and $\chi_{n, 1-\alpha}^2$ denotes the $(1-\alpha)$-th quantile of the $\chi^2_n$ distribution.}
\end{thm}

The proof of the theorem is given in Appendix \ref{sec:pfdistribution_A}. Note that, by Proposition \ref{ppn:mean_var}, the elements of 
$\mathrm{Cov}_{H_0}(\vec{\bm A}_N)/N$ has a non-degenerate limit, that is, there is a ${K\choose 2} \times {K \choose 2}$ matrix $\bm \Gamma$, such that $\mathrm{Cov}_{H_0}(\vec{\bm A}_N)/N \rightarrow \bm \Gamma$. An application of the Slutsky's theorem and \eqref{ant} then implies,  
$$\frac{\vec{\bm A}_N-\E_{H_0} \vec{\bm A}_N}{\sqrt N} \dto N_{{K \choose 2}}(0, \bm \Gamma),$$
and the test with rejection region 
\begin{align}\label{eq:SKNregion}
\left\{ \frac{\left(\vec{\bm A}_N - \mathbb{E}_{H_0} \vec{\bm A}_N\right)^\top \bm \Gamma^{-1} \left(\vec{\bm A}_N - \mathbb{E}_{H_0} \vec{\bm A}_N\right)}{N}>  \chi_{\binom{K}{2}, 1-\alpha}^2\right\},
\end{align}
is also asymptotically level $\alpha$.

Note that Theorem \ref{nullasmS} gives the asymptotic normality of entire cross-count vector, which implies the normality of any linear function of the cross-count vector, in particular, the MCM statistic \eqref{eq:RKN} as well.  More formally, \eqref{ant} implies, under $H_0$, 
\begin{equation}\label{eq:nullR2N}
Q_{K, N}:=\frac{R_{K, N}-\mathbb{E}_{H_0}(R_{K, N})}{\sqrt{\mathrm{Var}_{H_0}(R_{K, N})}} \dto N(0, 1),
\end{equation}
where, by Proposition \ref{ppn:mean_var}, $\mathbb{E}_{H_0}(R_{K, N}) = \frac{\sum_{s<t}N_sN_t}{N-1}$ and 
$$\mathrm{Var}_{H_0}(R_{K, N}) =  \frac{G_1}{N-1} \left(1- \frac{G_1}{N-1} \right) + \frac{G_1^2 - G_1 - 2G_2}{(N-1)(N-3)},$$ 
with $G_1 := \sum_{1 \leq s<t \leq K} N_s N_t$ and $G_2:= \frac{1}{2}\sum_{s=1}^K N_s(N-N_s)(N-N_s-1)$.  Therefore, the asymptotically level $\alpha$ test has rejection region $\{Q_{K, N} <  z_{\alpha}\}$, where $z_{\alpha}$ is the $\alpha$-th quantile of the standard normal distribution.

\begin{remark} (Non-Euclidean data) As the conditional (permutation) null distribution is same as the unconditional distribution (by Proposition \ref{ppn:H0_distribution}), it follows from the proof of Theorem \ref{nullasmS} that the asymptotic null distributions for $S_{K, N}$ and $Q_{K, N}$ obtained above hold verbatim for non-Euclidean spaces, as long as a similarity measure on the sample space can be defined. This is one of the highlights of tests based on inter-point distances, which makes them readily applicable for combinatorial and object data \cite{chfr,hcII}. Maa et al. \cite{interpoint}  provided theoretical motivations for using tests based on inter-point distances, by showing that, under mild conditions, two multivariate distributions are equivalent if and only if the distributions of inter-point distances within each distribution and between the distributions are equivalent. 
\end{remark}

\subsection{Consistency}
\label{sec:consistency}

In this section, we show the consistency of the tests discussed above. To this end,  assume that the $K$ distributions $F_1, F_2,\ldots, F_K$ have densities $f_1, f_2,\ldots, f_K$, respectively, with respect to the Lebesgue measure on $\mathbb{R}^d$. A test is said to be {\it universally consistent} for the hypothesis \eqref{eq:Ksample} if the power of the test converges to 1, whenever there exists $1 \leq s \ne t \leq K$ such that $f_s \ne f_t$ on a set of positive Lebesgue measure. Recently, Arias-Castro and Pelletier \cite{castropell} showed that the 2-sample CM test is universally consistent. Their arguments can be easily adapted to show the universal consistency of the MCM and MMCM tests:

\begin{thm}\label{consmain} In the usual limiting regime \eqref{eq:Nlimit}, $\frac{1}{N} \bm A_N \rightarrow \bm H=((h_{st}))_{1 \leq s, t \leq N}$ almost surely, where 
\begin{align}\label{eq:HPab}
h_{st}=   \left\{
\begin{array}{ll}
p_s p_t \int_{\mathbb{R}^d}\frac{f_s(z) f_t(z)}{\sum_{a=1}^K p_a f_a(z)}~\mathrm dz& \textrm{if}~s\neq t, \\\\
\frac{p_s^2}{2}\int_{\R^d} \frac{f_s^2(z)}{\sum_{a=1}^K p_a f_a(z)}~\mathrm dz & \textrm{otherwise}. 
\end{array} 
\right.
\end{align}
This implies that the MCM test with rejection $\{Q_{K, N} <  z_{\alpha}\}$ and the MMCM test with rejection region \eqref{eq:SKNregion} are universally consistent. 
\end{thm}

The proof of the theorem is given in Appendix \ref{app2}. The limit in \eqref{eq:HPab} implies that the MCM statistic (recall \eqref{eq:RKN})
\begin{align}\label{eq:R_consistency}
R_{K, N} \stackrel{a.s.} \rightarrow \sum_{1\leq s<t\leq K}h_{st} :=H(f_1,f_2,\ldots,f_K) := \tfrac{1}{2} - \mathrm{tr}(\bm H).
\end{align}
The consistency of the MCM and the MMCM tests then follows from the fact that 
$$H(f_1,f_2,\ldots,f_K) \leq H(f,f,\ldots,f),$$ and equality holds if and only if $f_1=f_2=\cdots=f_K$ outside a set of Lebesgue measure $0$ (details given Appendix \ref{app2}).

\begin{remark} (Henze-Penrose divergence) In the case $K=2$, the limiting constant $h_{12}$ equals $1-\delta(f_1, f_2)$, 
where $$\delta(f_1, f_2)=\int\frac{p_1^2 f_1^2(x)+p_2^2 f_2^2(x)}{p_1 f_1(x)+ p_2 f_2(x)} \mathrm d x,$$ is the well-known {\it Henze-Penrose divergence} between probability measures \cite{gyorfinemetz1}. This quantity appears as the almost sure limit of a large  class of graph-based 2-sample tests, which includes the Friedman-Rafsky test \cite{henzepenrose}, the nearest-neighbor based tests \cite{hn}, and the CM test \cite{castropell}, and has an interesting interpretation in terms of treatment-control assignment, using the propensity score \cite{covariate_balance}. For general $K$, the limit $h_{st}$ in \eqref{eq:HPab} is a multi-sample generalization of the Henze-Penrose integral, which aggregated over $ 1\leq s < t \leq K$ (as in the RHS of \eqref{eq:R_consistency}), is a global measure of dissimilarity between the densities $f_1, f_2, \ldots, f_K$.
\end{remark}

\section{Power Comparisons}\label{simulsec}

In this section, we illustrate the effectiveness of the tests described above by comparing their power with various other parametric and non-parametric tests, for several alternative hypotheses across different dimensions and number of groups. In Section \ref{sec:mfr} we illustrtate the advantage of using optimal matchings by comparing the performance of the MCM and MCMM tests with the multisample Friedman-Rafsky test (a natural generalization of the 2-sample Friedman-Rafsky test \cite{fr}, where the optimal matching is replaced with the minimum spanning tree (MST)).  In Section \ref{sec:parametric} we compare our matching based tests with other relevant parametric tests. Finally, in Section \ref{sec:comparison} we present an extensive comparison of the empirical power of the MCM and MCMM tests, across increasing dimensions ($d$) and group sizes ($K$). 
Additional simulations are given in Appendix \ref{sec:lognormalapp}.  Throughout, the nominal level of the tests are chosen to be $0.05$. 

\subsection{Comparison with the MST}\label{sec:mfr}
  
Here, we compare the performance of the optimal matching based tests described above, with the test based on the MST. As in \eqref{eq:RKN}, a natural extension of the 2-sample Friedman-Rafsky test based on the MST, is the \textit{multisample Friedman-Rafsky} test (MFRT), which rejects the null hypothesis for small values of the statistic: 
\begin{align}\label{eq:MKN}
T_{K, N}:= \sum_{1\leq s<t\leq K}  \sum_{i=1}^{N_s}\sum_{j=1}^{N_t}\bm 1 \{(X_i^{(s)}, X_j^{(t)} ) \in E(\cT(\mathscr{X})) \}, 
\end{align}
where $\cT(\sX)$ is the minimum spanning tree of the pooled sample $\mathscr{X} := \{\bm X^{(1)},\bm X^{(2)},\ldots,\bm X^{(K)} \}$. 
The MFRT, unlike the MCM and MMCM statistics, is not distribution-free under the null, however, it can be easily calibrated as a permutation test.  We compare the power of this test with the MCM and the MMCM tests in the following two scenarios. All the tests are calibrated using 500 permutations, and the empirical power is calculated over 500 iterations. 

\begin{figure}[h]
\includegraphics[width=5.1in]{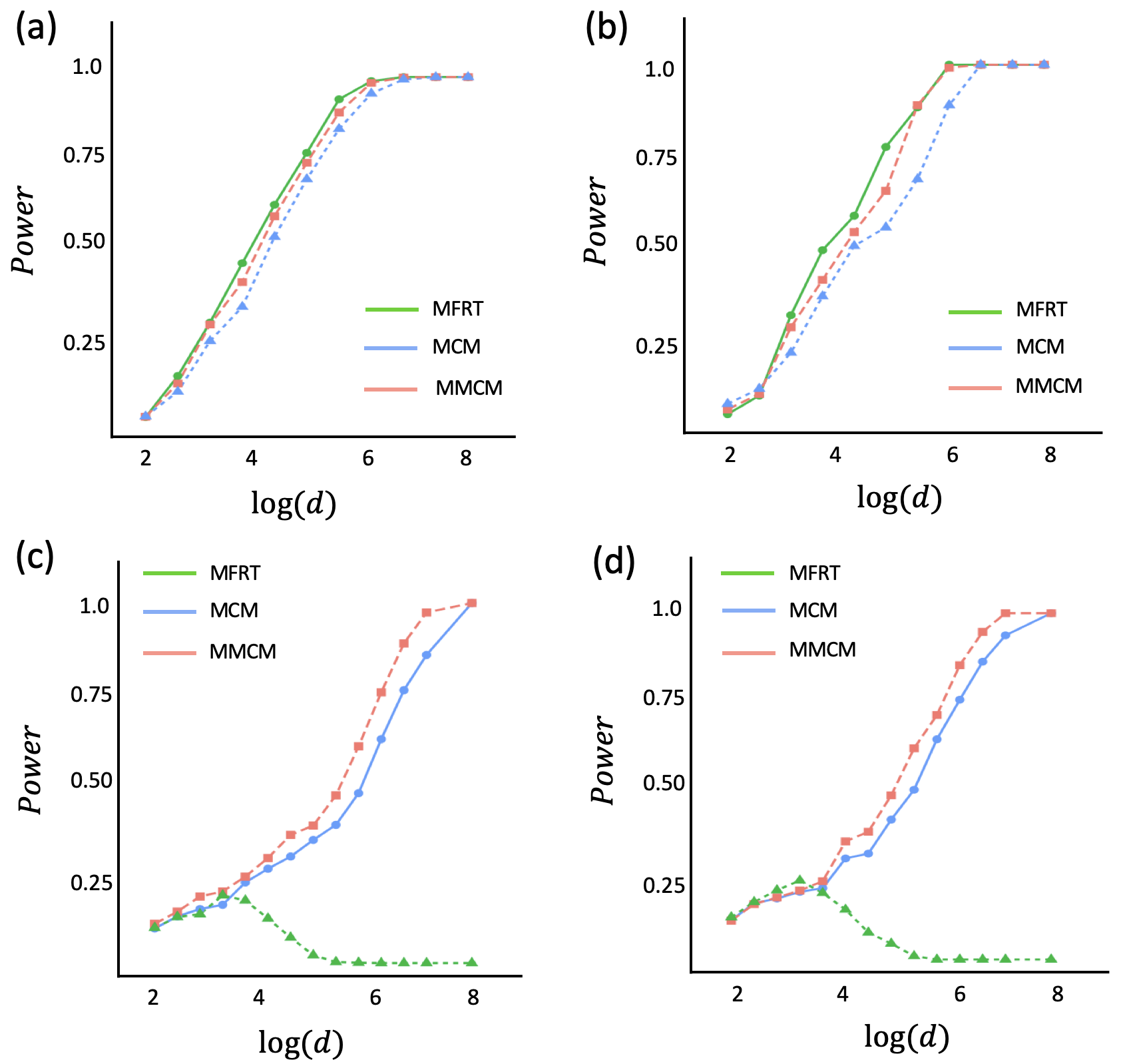}
\caption{\small{Power of the different tests across increasing dimension in (a) the normal location family with $K=3$ groups, (b) the normal location family with $K=4$ groups, (c) the spherical normal scale family with $K=3$ groups, and (d) the spherical normal scale family with $K=4$ groups.}}  
\label{fig:first}
\end{figure}

\begin{itemize}

\item {\it Normal Location}: Here, we consider the family $\{N_d(\bm \mu, \mathrm I): \bm \mu \in \R^d\}$. Figure \ref{fig:first}(a) shows the empirical power of the MFRT, the MCM test, and the MCMM test, when $K=3$ and the data consists of 100 samples each from $N_d(\bm 0, \mathrm I)$, $N_d(0.3 \cdot \bm 1, \mathrm I)$, and $N_d(0.6 \cdot \bm 1, \mathrm I)$, respectively.  Figure \ref{fig:first}(b) shows the empirical power of tests when $K=4$ and 100 samples each are drawn from $N_d(\bm 0, \mathrm I)$, $N_d(0.25  \cdot \bm 1, \mathrm I)$, $N_d(0.5 \cdot \bm 1, \mathrm I)$ and $N_d(0.75 \cdot \bm 1, \mathrm I)$, respectively. In both cases, the dimension $d$ varies from 2 to 2000.

\item {\it Spherical Normal Scale}: Here, we consider the family $\{N_d(\bm 0, \sigma^2 \mathrm I):  \sigma >0\}$. Figure \ref{fig:first}(c) shows the empirical power of the MFRT, the MCM test, and the MCMM test, when $K=3$ and the data consists of 100 samples each from $N_d(\bm 0, \mathrm I)$, $N_d(\bm 0, 1.5 \cdot \mathrm I)$, and $N_d(\bm 0, 2 \cdot\mathrm I)$, respectively.  Figure \ref{fig:first}(d) shows the empirical power of tests when $K=4$ and 100 samples each are drawn from $N_d(\bm 0, \mathrm I)$, $N_d(\bm 0, 1.25 \cdot \mathrm I)$, $N_d(\bm 0, 1.5 \cdot \mathrm I)$ and $N_d(\bm 0, 1.75 \cdot \mathrm I)$, respectively.  As before, the dimension $d$ varies from 2 to 2000.

\end{itemize}

The plots show that for location alternatives, the 3 tests have very similar power, with the MFRT performing marginally better than the MCMM, which is marginally better than the MCM. However, in the scale problem, the MCM and the MMCM tests drastically outperforms the MFRT. Here the power of the MFRT goes down to zero as the dimension increases, whereas the MCM and MCMM both have power improving with dimension and eventually going up to 1, illustrating the benefits of the distribution-free property of optimal matchings in high-dimensional problems.


\subsection{Comparison with Parametric Tests} 
\label{sec:parametric}

Here, we compare the performance of the MCM and the MMCM tests with baseline parametric tests, for relatively low-dimensional problems (where the corresponding parametric tests are applicable). As before, we consider location and scale alternatives in the normal family. Here, use the asymptotic distributions derived above to choose the cutoffs of the tests.

\begin{example}({\it Normal Location}) Here, we compare the MCM and  the MMCM tests with the Anderson's  test for Gaussian location alternatives \cite{anderson}.\footnote{This is a standard multisample method for testing difference of normal means. In the case where all the $K$ sample sizes are equal, the Anderson's test constructs, a vector $V^{(s)}$ of length $(K-1)d$ formed by appending $K-1$ linearly independent contrasts based on the $s$-th observations from each of the $K$ classes,  for each $1\leq s \leq N/K$. Then the test statistic is based on a Hoteling's $T^2$ statistic constructed from the vectors $V^{(1)},\ldots,V^{(N/K)}$. The reader is referred to \cite{anderson} for the precise description of this test.} As is common in parametric tests, since Anderson's statistic requires inverting certain sample covariance matrices, the test, unlike the MCM and the MMCM, is inapplicable for large dimensions, specifically if $d \geq \frac{N}{K(K-1)}$. When the dimension is much smaller than this threshold and there is a moderate sample size, Anderson's test  performs well (this is expected because the test is specifically designed for such alternatives). However, when the dimension increases and comes closer to the boundary,  Anderson's test start to lose power.  In fact, we see below that for dimensions close to this threshold, that non-parametric matching based methods outperform Anderson's test, for relatively small sample sizes. 
\begin{itemize}
\item Figure \ref{fig:parametric}(a) shows the empirical power (over 500 iterations) of the Anderson's test, the MCM test, and the MCMM test, for $K=7$ classes in dimension $d=16$. The horizontal axis shows the magnitude of separation $\Delta$, and the data consists of 100 samples from each of the following 7 distributions: for $1 \leq s \leq 7$, the $s$-th distribution corresponds to $N_{16}( (s-1)\Delta \cdot \bm 1, \mathrm I)$. Each simulation was repeated for 6 values of $\Delta$: $0.01,0.03,0.05,0.07,0.1,0.15$.

\begin{figure}
	\begin{center}
		\includegraphics[width=5.8in,height=6.5cm]{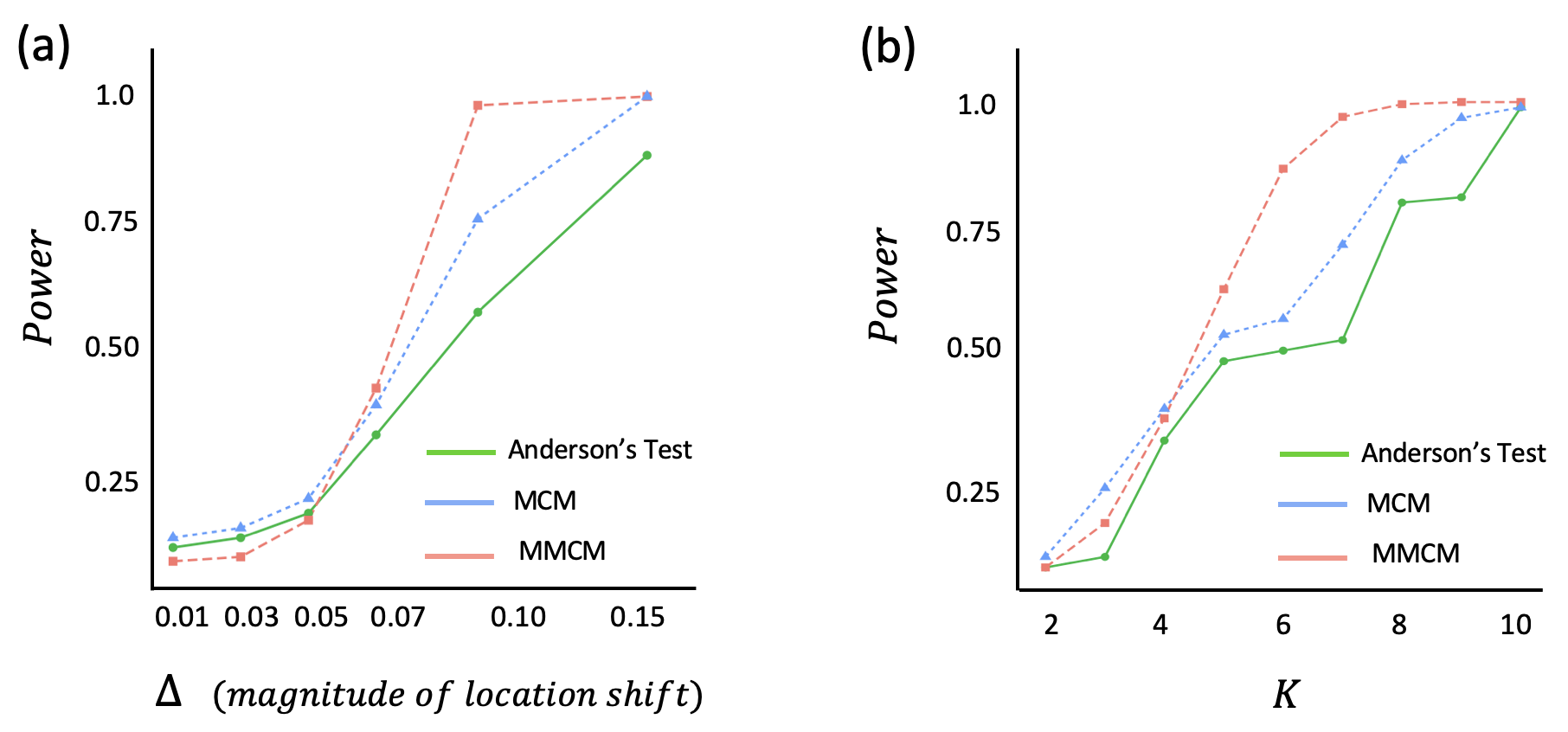}
	\end{center}
	\caption{\small{Power of the different tests in the normal location family when the number of classes (a) $K=7$, and (b) $K$ varying between $2$ and $10$.}} \vspace{-0.1in}\label{fig:parametric}
\end{figure}

\item Figure \ref{fig:parametric}(b) shows the empirical power when we vary both the number of classes $K$ and the dimension $d$. Here, $K$ varies from $2$ to $10$ (shown in the horizontal axis), and for each $K$ the dimension $d$ is chosen just below the dimension threshold of Anderson's test, and the $s$-th distribution corresponds to  $N_d(\frac{s-1}{10}\cdot \bm 1,  \mathrm I)$, for $1 \leq s \leq K$. 
\end{itemize} 
In both the cases, we observe that the power of the MCM and the MMCM tests are noticeably better than that of Anderson's test.  
\end{example}

\begin{example} ({\it Normal Scale}) Here, we compare the MCM and  the MMCM test with the likelihood ratio test (LRT)  for the equality of covariance matrices, when the means are unknown,\footnote{In this case, the LRT statistic rejects for small values of $\lambda = \exp(-\frac{1}{2}\sum_{i=1}^K N_i \log|\bm S_i^{-1}\bm S|)$, where $\bm S_i$ is the covariance matrix of the $i$-th sample, for $1\leq i \leq K$, and $\bm S$ is the pooled sample covariance matrix. Under the null hypothesis of equality of covariance matrices, $-2\log \lambda$ has an asymptotic chi-squared distribution with $\frac{1}{2}d(d+1)(K-1)$ degrees of freedom.} in the Normal scale family. This test performs well for small dimensions, especially when the sample sizes across the classes are equal. However, we see below that even in relatively small dimensions, the LRT performs poorly when the sample sizes become unbalanced, but the matching based tests continue to have significant power.  

\begin{itemize}

\item Figure \ref{fig:parsph}(a) shows the empirical power (over 500 iterations) of the LRT, the MCM test, and the MCMM test, for $K=4$ classes in dimension $d=20$. The horizontal axis shows the magnitude of separation $\Delta$, and the data consists of 80, 95, 110 and 125 samples from the following 4 distributions: for $1 \leq s \leq 4$, the $s$-th distribution corresponds to $N_{20}\left( \bm 0, (1+(s-1)\Delta) \mathrm I\right)$. Each simulation was repeated for $10$ values of $\Delta$: $0.4,0.45,0.5,0.55,0.6,0.65,0.7,0.75,0.8,0.85$.   
		
\item Figure \ref{fig:parsph}(b) shows the empirical power (over 500 iterations) when we vary the sample size difference $\delta$ among the classes from $8$ to $18$. The horizontal axis shows this sample size difference $\delta$. For each $\delta$, the data consists of $80,80+\delta, 80+2\delta$ and $80+3\delta$ samples from the the following $4$ distributions respectively: for $1 \leq s \leq 4$, the $s$-th distribution corresponds to $N_{20}\left( \bm 0, \frac{s+1}{2}\mathrm I\right)$.
\end{itemize} 
In the first case, we observe that the power of all the tests increase to 1, but the MCM and the MMCM tests dominate that of the LRT by a significant margin for smaller separations. In the second case, however,  with increase in the difference of sample size across the classes, the power of the LRT decreases to $0$, while the power of the MCM and MMCM tests remain stable at high values, illustrating the robustness of these methods even for low-dimensional problems. 
\end{example}

\begin{figure}
	\begin{center}
		\includegraphics[width=5.8in,height=6.5cm]{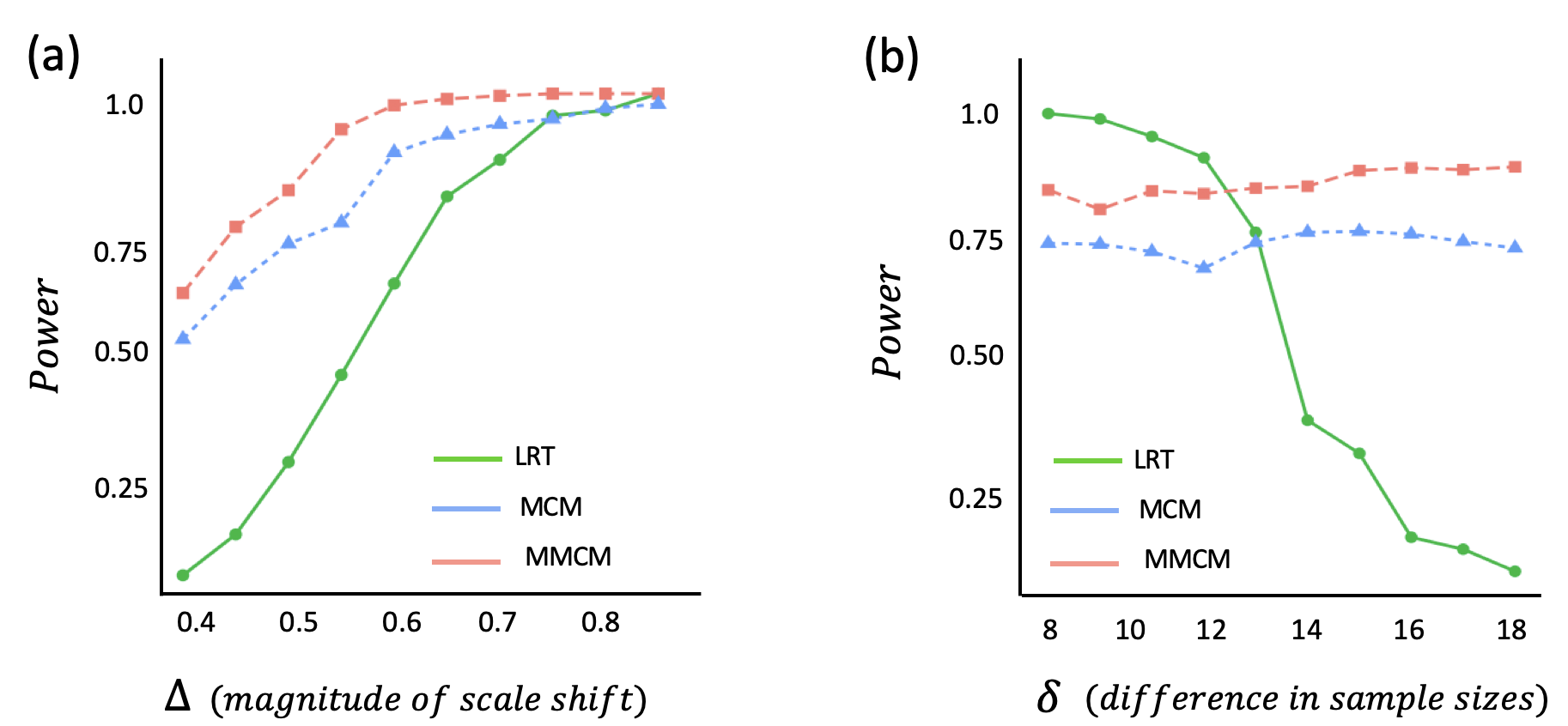}
	\end{center}
	\caption{\small{Power of the different tests for the spherical normal scale problem for (a) varying scale shift, and (b) varying difference in sample size}} \label{fig:parsph}\vspace{-0.1in}
\end{figure}

\subsection{Comparison between the MCM and the MMCM Tests}
\label{sec:comparison}

In this section, we compare the finite-sample power of the MCM and the MMCM tests in various examples.  Here, we consider 3 distributional models: (1) the normal location family $\{N_d(\bm \mu, \mathrm I): \bm \mu \in \R^d\}$, (2) the spherical normal scale family $\{N_d(\bm 0, \sigma^2 \mathrm I): \sigma > 0\}$, and (3) the  equi-correlated normal scale family $\{N_d(\bm 0, (1-\rho) \mathrm I) + \rho \bm 1 \bm 1^\top: 0 \leq \rho < 1\}$.   A typical simulation instance looks as follows. We generate samples from $K$ different $d$-dimensional distributions from an underlying distributional model. The difference between the $K$ distributions is quantified by a separation parameter $\Delta$ (which depends on location parameter/spherical scale parameter/correlation parameter, depending on the underlying distributional model).  For each of the distributional models we consider two scenarios: (1) the \textit{fixed class scenario} where the number of classes $K$ is fixed and we perform a two-way power comparison with $\Delta$ versus $d$, and (2) the \textit{fixed dimension scenario}, where we fix the dimension $d$, and perform a two-way power comparison with $\Delta$ versus $K$. 
In all the simulations, the power is calculated over 100 iterations.  Additional simulations, comparing the MCM and the MCMM tests, in the lognormal family are given in Appendix \ref{sec:lognormalapp}. Overall, we observe that MCM and MCMM tests are comparable for small dimension and group sizes, but the MCMM test outperforms the MCM as the dimension, groups, or separation increases.

\begin{table}[t]
\centering
\begin{minipage}[c]{0.59\textwidth}
\centering
\small{
\begin{table}[H]
\begin{tabular}{c|c||ccccccc}
    \hline
    $\Delta\downarrow$ & Dimension &  5 & 10 & 50 & 100& 200 & 300 & 500 \\  
    \hline
    \hline
    \multirow{2}{*}{.04}&MCM & .26 & .30 & {\bf .16} & {\bf .19} & {\bf .36} & .41 & .57 \\
     &MCMM & \textbf{.29}  & \textbf{.44} & .10 & .17 & .25 & {\bf .44} & {\bf .76} \\
    \hline
    \multirow{2}{*}{.06}&MCM & .43 & .51 & .31 & {\bf .43} & {\bf .41} & .52 & .92 \\
     &MCMM & {\bf .59}  & {\bf .69} & {\bf.46} & .32 & .37 & {\bf .70} &  {\bf 1.0} \\
    \hline
    \multirow{2}{*}{.08}&MCM & .49 & .61 & .39 & .48 & .50 & .85 & 1.0 \\
     &MCMM & {\bf .67}  & {\bf .81} & {\bf .65} & {\bf .64} & {\bf .65} & {\bf .99} & {\bf 1.0} \\
    \hline
    \multirow{2}{*}{.10}&MCM & .66 & .70 & .77 & .60 & .97 & .99 & 1.0 \\
     &MCMM & {\bf .81}  & {\bf .84} & {\bf .90} & {\bf .79} & {\bf .99} & {\bf 1.0} & {\bf 1.0} \\
    \hline
    \multirow{2}{*}{.12}&MCM & .81 & .97 & .81 & .87 & 1.0 & 1.0 & 1.0 \\
     &MCMM & {\bf .94}  & {\bf 1.0} & {\bf .93} & {\bf 1.0} &  {\bf 1.0} & {\bf 1.0} & {\bf 1.0} \\
    \hline
\end{tabular}
\end{table}
\vspace{-0.1in}
(a)
}
\end{minipage}
\begin{minipage}[c]{0.39\textwidth}
\centering
\small{
\begin{table}[H]
\begin{tabular}{c|c||cccc}
    \hline
    $\Delta\downarrow$ & Groups & 4  & 6 & 8 & 10  \\  
    \hline
    \hline
    \multirow{2}{*}{.04}&MCM & {\bf .06} & .42 & .45 & .80  \\
     &MCMM & .04  & {\bf .53} & {\bf .78} & {\bf .97} \\
    \hline
    \multirow{2}{*}{.05}&MCM & {\bf .11} & .61 & .85 & 1.0  \\
     &MCMM & .07  & {\bf .78} & {\bf .99} & {\bf 1.0} \\
    \hline
    \multirow{2}{*}{.07}&MCM & {\bf .19} & .77 & 1.0 & 1.0  \\
     &MCMM & .13  & {\bf 0.96} & {\bf 1.0} & {\bf 1.0}  \\
    \hline
    \multirow{2}{*}{.09}&MCM & .47 & .93 & 1.0 & 1.0  \\
     &MCMM & {\bf .53}  & {\bf 1.0} & {\bf 1.0} & {\bf 1.0} \\
    \hline
    \multirow{2}{*}{.10}&MCM & .55 & 1.0 & 1.0 & 1.0 \\
     &MCMM & {\bf .70}  & {\bf 1.0} & {\bf 1.0} & {\bf 1.0}  \\
    \hline
\end{tabular}
\end{table}
}
\vspace{-0.1in}
(b)
\end{minipage}
\caption{\small{Power of the MCM and the MMCM tests in the normal location family with (a) the number of classes $K=6$ fixed, and (b) the dimension $d=150$ fixed. (The higher power in each case is in bold.)}}
\label{table:normallocation}
\vspace{-0.1in}
\end{table}

\normalsize	

\begin{itemize}

\item {\it Normal Location}: Here, we consider samples from the following $K$ distributions: $N_d((s-1)\Delta\bm 1, \mathrm I)$, for $1 \leq s \leq K$. Table \ref{table:normallocation}(a) shows the fixed class scenario, where we take $K=6$ groups and vary the dimension $d$ from $5$ to $500$, and $\Delta$ from $0.04$ to $0.12$. Table \ref{table:normallocation}(b) shows the fixed dimension scenario, where the dimension $d = 150$ is fixed, the number of groups $K$ varies along $4,6,8,10$, and $\Delta$ varies from  $0.04$ to $0.10$. In both cases, the sample sizes  were taken in equal increments of $50$, starting from $50$.

\item {\it Spherical Normal Scale}: Here, we consider samples from the following $K$ distributions: $N_d(\bm 0, (1+(s-1)\Delta) \mathrm I)$, for $1 \leq s \leq K$. Table \ref{table:normalscale}(a) shows the  fixed class scenario, with $K=6$ and dimension $d$ varying from $5$ to $500$, and $\Delta$ varying from $0.05$ to $0.4$. Table \ref{table:normalscale}(b) shows the  fixed dimension scenario, where the $d = 150$ is fixed, and $K$ varies along $4,6,8,10$, and $\Delta$ varies from $0.05$ to $0.4$. As before, in both cases, the sample sizes  were taken in equal increments of $50$, starting from $50$.

\begin{table}[h]
	\centering
	\begin{minipage}[c]{0.59\textwidth}
		\centering
		\small{
			\begin{table}[H]
				\begin{tabular}{c|c||ccccccc}
					\hline
					$\Delta\downarrow$ & Dimension &  5 & 10 & 50 & 100& 200 & 300 & 500 \\  
					\hline
					\hline
					\multirow{2}{*}{.15}&MCM & \bf .12 & \bf .16 & .22 & .33 & .51 & .55 & .71 \\
					&MMCM & .06  & .13 &\bf  .28 &\bf  .37 &\bf  .67 &\bf  .86 &\bf  .92\\
					\hline
					\multirow{2}{*}{.20}&MCM & \bf .13 & \bf .24 & .41 & .46 & .68 & .79 & .94 \\
					&MMCM & .12  & .19 &\bf  .59 &\bf  .78 & \bf .95 & \bf .99 &  \bf 1.0\\
					\hline
					\multirow{2}{*}{.25}&MCM & .21 & .31 & .47 & .60 & .81 & .88 & 1.0 \\
					&MMCM & \bf .23  & \bf .38 &\bf  .85 & \bf .98 & \bf 1.0 & \bf 1.0 & \bf 1.0 \\
					\hline
					\multirow{2}{*}{.30}&MCM & .22 & .43 & .70 & .91 & .98 & 1.0 & 1.0 \\
					&MMCM & \bf .25  & \bf .53 & \bf .99 & \bf 1.0 & \bf 1.0 &\bf  1.0 & \bf 1.0\\
					\hline
					\multirow{2}{*}{.35}&MCM & .24 & .38 & .75 & .87 & 1.0 & 1.0 & 1.0 \\
					&MMCM & \bf .30  & \bf .54 & \bf .99 & \bf 1.0 & \bf 1.0 & \bf 1.0 & \bf 1.0 \\
					\hline
					\multirow{2}{*}{.40}&MCM & .29 & .50 & .86 & 1.0 & 1.0 & 1.0 & 1.0 \\
					&MMCM & \bf .49  & \bf .85 & \bf 1.0 & \bf 1.0 &  \bf 1.0 & \bf 1.0 & \bf 1.0 \\
					\hline
				\end{tabular}
			\end{table}
			\vspace{-0.1in}
			(a)
		}
	\end{minipage}
	\begin{minipage}[c]{0.39\textwidth}
		\centering
		\small{
			\begin{table}[H]
				\begin{tabular}{c|c||cccc}
					\hline
					$\Delta\downarrow$ & Groups & 4  & 6 & 8 & 10  \\  
					\hline
					\hline
					\multirow{2}{*}{.15}&MCM & \bf .27 & .68 & .91 & 1.0  \\
					&MMCM & .18  & \bf .91 & \bf 1.0 & \bf 1.0 \\
					\hline
					\multirow{2}{*}{.20}&MCM & .48 & .89 & .99 & 1.0  \\
					&MMCM & \bf .48  & \bf 1.0 & \bf 1.0 & \bf 1.0 \\
					\hline
					\multirow{2}{*}{.25}&MCM & .66 & .96 & 1.0 & 1.0  \\
					&MMCM & \bf .89  & \bf 1.0 & \bf 1.0 & \bf 1.0  \\
					\hline
					\multirow{2}{*}{.30}&MCM & .87 & 1.0 & 1.0 & 1.0  \\
					&MMCM & \bf .98  & \bf 1.0 & \bf 1.0 & \bf 1.0 \\
					\hline
					\multirow{2}{*}{.35}&MCM & .95 & 1.0 & 1.0 & 1.0 \\
					&MMCM & \bf 1.0  & \bf 1.0 & \bf 1.0 & \bf 1.0  \\
					\hline
					\multirow{2}{*}{.40}&MCM & 1.0 & 1.0 & 1.0 & 1.0 \\
					&MMCM & \bf 1.0  & \bf 1.0 & \bf 1.0 & \bf 1.0  \\
					\hline
				\end{tabular}
			\end{table}
		}
		\vspace{-0.1in}
		(b)
	\end{minipage}
	\caption{\small{Power of the MCM and the MMCM tests in the spherical normal scale family with (a) the number of classes $K=6$ fixed, and (b) the dimension $d=150$ fixed.}}
	\vspace{-0.15in}
\label{table:normalscale}		
\end{table}

\normalsize

\item {\it Equi-correlated Normal Scale}:  Here, we consider samples from the following $K$ distributions: $N_d(0, (1-\rho_s)\mathrm I + \rho_s \bm 1 \bm 1^\top)$, where $\rho_s := (s-1)\frac{\Delta}{K-1}$, for $1 \leq s \leq K$.   Table \ref{table:normalequicorrelationscale}(a) shows the  fixed class scenario, with $K=6$ and dimension $d$ varying from $5$ to $500$, and $\Delta$ varying from $0.15$ to $0.4$. The sample sizes are taken to be $50, 100, 150, 200, 250$ and $300$.
Table \ref{table:normalequicorrelationscale}(b) shows the  fixed dimension scenario, where the $d = 150$ is fixed, and $K$ varies along $4,6,8,10$, and $\Delta$ varies from $0.15$ to $0.4$, as before. The sample sizes are taken in equal increments from $50$ to $200$ when $K=4$, from $50$ to $300$ when $K=6$, from $50$ to $260$ when $K=8$, and from $50$ to $230$ when $K=10$.

\end{itemize}

In all the simulations above (and those in Appendix \ref{sec:lognormalapp}), we observe that the power of both the MCM and the MMCM tests improve with increasing separation and dimension. For smaller dimensions and separations, the power of both the tests are comparable, however, the MCMM test quickly gains power and performs noticeably better than the MCM, for higher dimensions and larger separations. This leads to our preference for using the MMCM over the MCM especially in high dimensions, which is often the case for real life datasets.

\begin{table}[h]
	\centering
	\begin{minipage}[c]{0.59\textwidth}
		\centering
		\small{
			\begin{table}[H]
				\begin{tabular}{c|c||ccccccc}
					\hline
					$\Delta\downarrow$ & Dimension &  5 & 10 & 50 & 100& 200 & 300 & 500 \\  
					\hline
					\hline
					\multirow{2}{*}{.15}&MCM & \bf .10 & .11 & .16 & .22 & .34 & \bf .38 & .35 \\
					&MMCM & .08  & \bf .11 & \bf .19 & \bf .27 & \bf .36 & .37 & \bf .39\\
					\hline
					\multirow{2}{*}{.20}&MCM & \bf .07 & \bf .10 & \bf .17 & .24 & .25 & .35 & .43 \\
					&MMCM & .06  & .09 & .15 & \bf .25 & \bf .38 & \bf .49 & \bf .64\\
					\hline
					\multirow{2}{*}{.25}&MCM & \bf .09 & \bf .16 & .25 & .27 & .39 & .44 & .50 \\
					&MMCM & .02  & .09 & \bf .29 & \bf .42 & \bf .55 &\bf .66 &\bf .78 \\
					\hline
					\multirow{2}{*}{.30}&MCM & .08 & \bf .17 & .27 & .32 & .56 & .67 & .71 \\
					&MMCM & \bf .14  & .13 & \bf .36 & \bf .48 & \bf .78 & \bf .79 & \bf .93\\
					\hline
					\multirow{2}{*}{.35}&MCM & \bf .12 & \bf .19 & .45 & .46 & .60 & .65 & .88 \\
					&MMCM & .06  & .15 & \bf .50 & \bf .70 & \bf .81 & \bf .90 & \bf .99 \\
					\hline
					\multirow{2}{*}{.40}&MCM & \bf .15 & .26 & .51 & .70 & .77 & .84 & .95 \\
					&MMCM & .12  & \bf .26 & \bf .70 & \bf .91 &  \bf .98 & \bf 1.0 & \bf 1.0 \\
					\hline
				\end{tabular}
			\end{table}
			\vspace{-0.1in}
			(a)
		}
	\end{minipage}
	\begin{minipage}[c]{0.39\textwidth}
		\centering
		\small{
			\begin{table}[H]
				\begin{tabular}{c|c||cccc}
					\hline
					$\Delta\downarrow$ & Groups & 4  & 6 & 8 & 10  \\  
					\hline
					\hline
					\multirow{2}{*}{.15}&MCM & \bf .26 & .19 & .29 & .31 \\
					&MMCM & .20  & \bf .26 & \bf .32 & \bf .34 \\
					\hline
					\multirow{2}{*}{.20}&MCM & \bf .43 & .30 & .29 & .18  \\
					&MMCM & .37  & \bf .35 & \bf .42 & \bf .26 \\
					\hline
					\multirow{2}{*}{.25}&MCM & .48 & .40 & .34 & \bf .32  \\
					&MMCM & \bf .54  & \bf .47 & \bf .43 & .28  \\
					\hline
					\multirow{2}{*}{.30}&MCM & \bf .61 & .53 & .43 & .28  \\
					&MMCM & .59  & \bf .63 & \bf .58 & \bf .47 \\
					\hline
					\multirow{2}{*}{.35}&MCM & .66 & .63 & .47 & .43 \\
					&MMCM & \bf .78  & \bf .80 & \bf .76 & \bf .71  \\
					\hline
					\multirow{2}{*}{.40}&MCM & .87 & .72 & .67 & .57 \\
					&MMCM & \bf .99 & \bf .95 & \bf .90 & \bf .83  \\
					\hline
				\end{tabular}
			\end{table}
		}
		\vspace{-0.1in}
		(b)
	\end{minipage}
	\caption{\small{Power of the MCM and the MMCM tests in the equi-correlated normal scale family with (a) the number of classes $K=6$ fixed, and (b) the dimension $d=150$ fixed.}}
\label{table:normalequicorrelationscale}	
\end{table}\vspace{-0.15in}

\normalsize	
 
\section{Distribution Under the Alternative}\label{condasm1}

In this section we will prove a central limit theorem for the vector of cross-counts, and, as a corollary, derive the asymptotic distribution of the cross-match \eqref{eq:R2N}, the MCM \eqref{eq:RKN} and MMCM \eqref{eq:SKN} statistics, under general alternatives. We begin with an alternative way to describe the joint distribution of the data $\{\bm X^{(1)},\bm X^{(2)},\ldots,\bm X^{(K)} \}$:
\begin{itemize}

\item Let $Z_1, Z_2, \ldots, Z_N$ be i.i.d. from the density $\phi_N:=\sum_{s=1}^K \frac{N_s}{N} f_s$ in $\R^d$, where $f_1, f_2, \ldots, f_K$ are the densities  (with respect to the Lebesgue measure on $\R^d$) of the distributions $F_1, F_2, \ldots, F_K$, respectively.

\item Given $\cZ_N=(Z_1, Z_2, \ldots, Z_N)$, assign a random label $L_j \in [K]:=\{1, 2, \ldots, K\}$ to $Z_j$, independently for each $1 \leq j \leq N$, where 
\begin{align}\label{eq:labels}
\P(L_j=s|Z_j)  = \frac{\frac{N_s}{N} f_s(Z_j)}{\phi_N(Z_j)}, \quad  \textrm{for all}~s \in [K].
\end{align}

\item Denote by $\eta_s := \sum_{i=1}^N \bm 1 \{L_i = s\}$, the number of elements labelled $s$. Then it is easy to verify that the joint distribution of $\left(\{Z_j: L_j=1\}, \{Z_j: L_j=2\}, \ldots, \{Z_j: L_j=K\}\right)$ conditional on $(\eta_1,\cdots,\eta_K)=(N_1, N_2, \ldots, N_K)$ is same as the joint distribution of the data $\left(\bm X^{(1)},\bm X^{(2)},\ldots,\bm X^{(K)} \right)$ (see Lemma \ref{condrel}).\footnote{Note that under the null, \eqref{eq:labels} simplifies to $P(L_j=s|Z_j)  = \frac{N_s}{N}$, and the procedure described above, is precisely the way to generate the permutation null distribution.} 
\end{itemize}
It is also often convenient (because of the independence of the labelings) to work with unconditional distribution of $(\{Z_j: L_j=1\}, \{Z_j: L_j=2\}, \ldots, \{Z_j: L_j=K\})$, which we will refer to as the {\it bootstrap alternative distribution}. 

Now, define the $K\times K$ matrix $\bm B_N=(b_{st})_{1 \leq s, t \leq K}$ as follows:
\begin{equation}\label{eq:newbdefn}
b_{st} =   \left\{
\begin{array}{ll}
\sum_{1\leq i\neq j\leq N}e(Z_i,Z_j) \bm 1 \{L_i=s, L_j = t\} & \textrm{if}~s\neq t, \\ \\ 
\frac{1}{2}\sum_{1\leq i\neq j\leq N} e(Z_i, Z_j) \bm 1 \{L_i=L_j = s\} & \textrm{if}~s= t.
\end{array} 
\right. 
\end{equation}
where $e(x,y) := \bm 1 \left\{(x,y) \in  E(\sg(\cZ_N \cup \{x, y\}))\right\}$. Moreover, for notational convenience, denote $\bm \eta := (\eta_1,\ldots,\eta_K)$ and $\bm N := (N_1,\ldots,N_K)$. Then, for $1 \leq s , t \leq K$ the cross/pure counts $a_{st}$ (recall \eqref{eq:st_count} and \eqref{eq:ss_count}) can be re-written in terms of the $\cZ_N$ and the labelings as follows: 
\begin{align}\label{eq:acountsZ_I}
a_{st} \stackrel{D}= b_{st}\Big|\{ \bm \eta = \bm N\},
\end{align}
Therefore, the conditional mean of the pure/cross-counts under the bootstrap alternative distribution is
\begin{align}\label{eq:muNst}
\nu_N(s, t):=\E_{H_1}(b_{st}|\cZ_N)=
\left\{
\begin{array}{ccc}
\sum_{1\leq i\neq j\leq N} e(Z_i, Z_j) h_{st}^{(N)}(Z_i, Z_j),   &   \text{if }  s \ne t \\ \\ 
\frac{1}{2}\sum_{1\leq i\neq j\leq N} e(Z_i, Z_j) h_{ss}^{(N)}(Z_i, Z_j),  &   \text{if }  s = t,
\end{array}
\right.
\end{align}
where $h_{st}^{(N)}(x, y)=\frac{N_s}{N}  \frac{N_t}{N}  \frac{f_s(x) f_t(y) }{\phi_N(x) \phi_N(y)}$. We denote the matrix of these conditional expectations by $\bm \mu_N=((\mu_N(s, t)))_{1 \leq s, t \leq K}$. Note the expressions in the RHS above is permutation invariant and a function of the pooled sample (forgetting the labels). For example, for $s \ne t$,
\begin{align}\label{eq:muNst_sample}
\nu_N(s, t) \Big| \{\bm \eta = \bm N\}\stackrel{D}=\sum_{1 \leq a, b \leq K} \sum_{i=1}^{N_a}\sum_{j=1}^{N_b} h_{st}^{(N)}(X_i^{(a)}, X_j^{(b)}) e(X_i^{(a)}, X_j^{(b)} ) =: \mu_N(s,t),
\end{align} 
which can be computed from the pooled data at a known alternative point $(f_1, f_2, \ldots, f_K)$.   As usual, denote by $\vec{\bm \mu}_N$ the vector obtained by concatenating the rows of the matrices $\bm \mu_N := ((\mu_N(s,t)))$ in the upper triangular part. Note that $\vec{\bm \mu}_N$ can be thought of as the conditional mean of the vector $\vec{\bm A}_N$ given the pooled sample, where the randomness comes only from the labeling of the classes. 

In the theorem below we show that the vector $\vec{\bm A}_N$ (recall that this is the  vector obtained by concatenating the rows of $\bm A_N$ in the upper triangular part, as defined in \eqref{eq:SKN})  centered by the corresponding vector of conditional means  $\vec{\bm \mu}_N$ and scaled appropriately, converges in distribution to a ${K \choose 2}$-dimensional multivariate normal in the usual asymptotic regime \eqref{eq:Nlimit}. The proof of the theorem is given in Appendix \ref{secalt}.

\begin{thm}\label{altclt} 
Under general alternatives, in the usual asymptotic regime \eqref{eq:Nlimit},   
\begin{align}\label{eq:clt_H1}
\frac{\vec{\bm A}_N - \vec{\bm \mu}_N}{\sqrt N}  \xrightarrow{D} N_{{K\choose 2}}(0, \bm \Gamma_{f_1, f_2, \ldots, f_k}),
\end{align}
where the covariance matrix $\bm \Gamma_{f_1, f_2, \ldots, f_k}$ is as in Definition \ref{defn:GammaK} (in Appendix \ref{secalt}). 
\end{thm}

The joint normality of the vector $\vec{\bm A}_N$ implies the normality of linear functions of $\vec{\bm A}_N $, in particular the MCM statistic (recall \eqref{eq:RKN}), which can be re-written as $R_{K, N}= \bm 1^\top \vec{\bm A}_N$. Therefore, 
\begin{align*}
\frac{R_{K, N} - \bm 1^\top \vec{\bm \mu}_{N}}{\sqrt N} \xrightarrow{D} N_{{K\choose 2}}(\bm 0, \bm 1^\top\bm \Gamma_{f_1, f_2, \ldots, f_K}\bm 1).
\end{align*}
Similarly, for the MMCT statistic, \eqref{eq:clt_H1} implies
$$\frac{(\vec{\bm A}_N - \vec{\bm \mu}_N)^\top \bm \Gamma_{f_1, f_2, \ldots, f_K}^{-1} (\vec{\bm A}_N - \vec{\bm \mu}_N) }{N} \xrightarrow{D} \chi_{\binom{K}{2}}^2.$$ 
Even though the general expression for the covariance matrix $\bm \Gamma_{f_1, f_2, \ldots, f_k}$ (Definition \ref{defn:GammaK} in Appendix \ref{secalt}) can be complicated, it simplifies nicely for the case $K=2$.  To this end, recall the  2-sample cross match statistic  $R_{2, N}$ from \eqref{eq:R2N}. Let $\mathscr{X}$ denote the pooled sample. Then, by \eqref{eq:muNst_sample}, 
\begin{align}\label{eq:cm_mean}
\E_{H_1}(R_{2, N}|\mathscr{X})=\sum_{1 \leq a, b \leq 2} \sum_{i=1}^{N_a}\sum_{j=1}^{N_b} \frac{N_1}{N}  \frac{N_2}{N}  \frac{f_1(X_i^{(a)}) f_2(X_i^{(b)}) }{\phi_N(X_i^{(a)}) \phi_N(X_i^{(b)})} e(X_i^{(a)}, X_j^{(b)} ).
\end{align}
We  now have the following result for the 2-sample cross match test, which is a straightforward calculation from \eqref{eq:clt_H1} above.

\begin{figure*}[h]
\begin{minipage}[l]{0.49\textwidth}
\hspace{-0.4in}
\begin{center} 
\small{(a)}
\end{center}
\includegraphics[width=3in,height=2.2in]
    {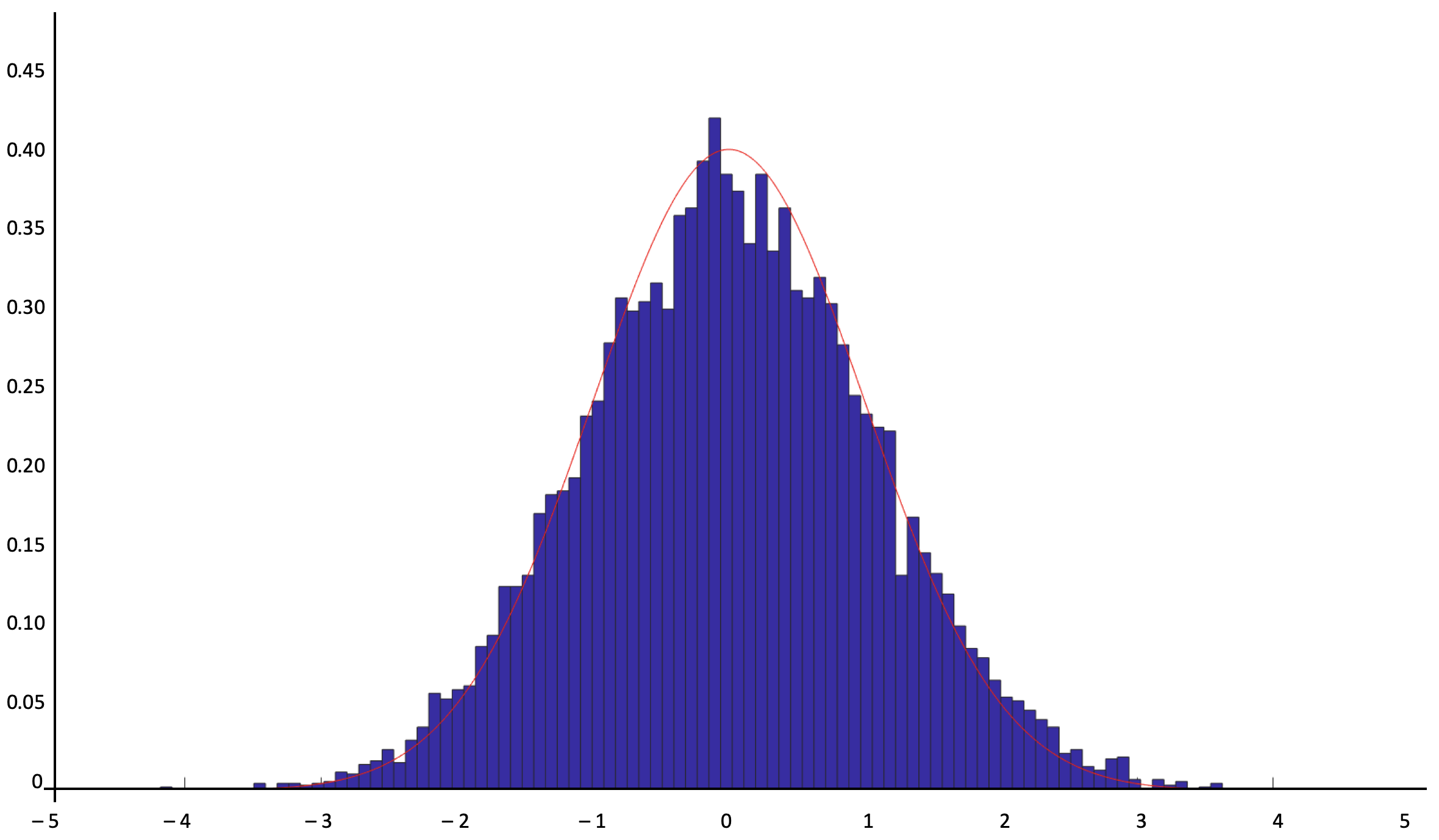}\\
\vspace{-0.4in}

\end{minipage}
\begin{minipage}[l]{0.49\textwidth}
\hspace{-0.25in}
\begin{center} 
\small{(b)}
\end{center}
\includegraphics[width=3in,height=2.2in]
    {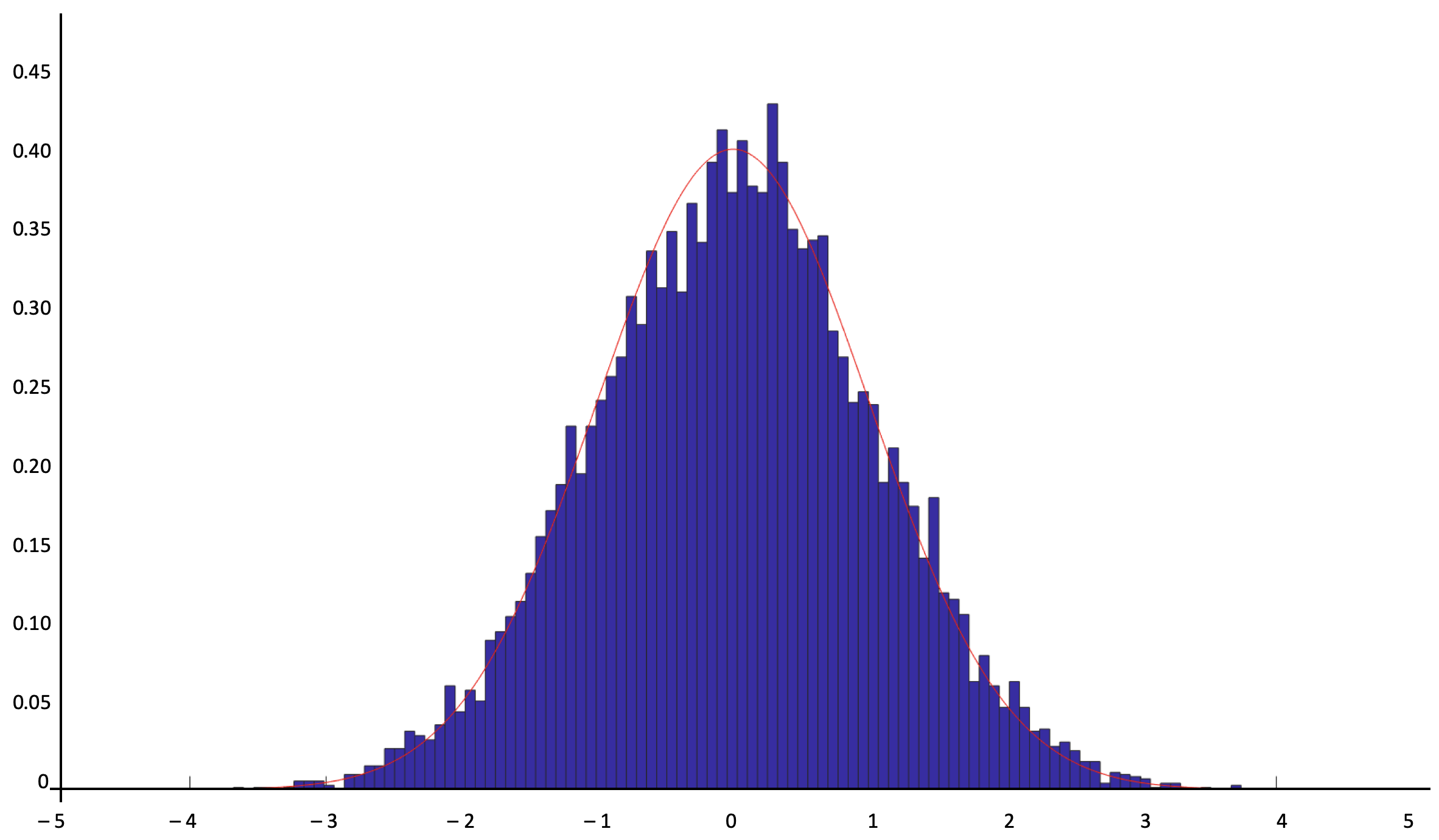}\\
\vspace{-0.4in}
\end{minipage}
\bigskip
\caption{\small{Histogram of the centered CM statistic (over 1000 iterations) and the predicted normal density (the red curve) when (a) $f_1=N_{10}(\bm 0, \mathrm I)$ and $f_2=N_{10}(\bm 1, \mathrm I)$ (normal location alternatives) and (b) $f_1=N_{10}(\bm 0, \mathrm I)$ and $f_2=N_{10}(\bm 0, 2\mathrm I)$ (normal scale alternatives).}}
\label{fig:crossmatch_H1}
\end{figure*}

\begin{cor}\label{rosenalt} For 2-sample cross match statistic $R_{2, N}$ as in \eqref{eq:R2N},  as $N \rightarrow \infty$, 
$$\frac{R_{2, N}- \E_{H_1}(R_{2, N}|\mathscr{X})}{\sqrt N} \xrightarrow{D} N(0, \gamma_{f_1, f_2}^2), $$
where $\E_{H_1}(R_{2, N}|\mathscr{X})$ is as in \eqref{eq:cm_mean} and 
$$\gamma_{f_1, f_2}^2 := p_1 p_2\left\{\int \frac{f_1(z)f_2(z) (p_1^2 f_1^2(z)+p_2^2 f_2^2(z))}{\phi(z)^3 } \mathrm d z - \left( \int \frac{f_1(z)f_2(z) (p_2 f_2(z)-p_1 f_1(z))}{\phi(z)^2}  \mathrm d z \right)^2\right\},$$
with $\phi:=p_1 f_1 + p_2 f_2$. 
\end{cor}



The results above add to our mathematical understanding of the alternative properties of matching-based tests, which up till now have been largely unexplored. Figure \ref{fig:crossmatch_H1} shows the histogram of the centered CM statistic computed using 600 samples from one distribution and 400 samples from another, repeated over 10000 iterations, for normal location and scale alternatives, and density of the corresponding limiting normal distribution (the red curve), as predicted by corollary above. The plots validate the asymptotic results and show that the normal approximation is quite accurate even for  moderate sample sizes.

\section{Application to Single Cell RNA Sequencing Data}
\label{sec:applications}

In this section we apply the tests described above to  single cell transcriptomics data obtained from the peripheral blood, cancer tissue and tumor-adjacent normal tissue of human subjects with hepatocellular carcinoma and non-small-cell lung cancer. Our goal is to investigate how biochemical metabolic pathways change across immune cells in a cancer environment, depending on the location of the tissue. We begin with a short background on biological pathways and the single cell RNA sequencing data.

Every tissue in the human body comprises of numerous different cell types, and each cell in turn contains tens of thousands of genes. The function of a tissue or an organ is rarely driven by a single unique gene, and analogously, complex disorders of organ dysfunction affect multiple genes. Therefore, to understand complex diseases, a systems biology approach examines sets or functional modules of related genes, called biological pathways. Because diseases such as cancer result from different combinations of perturbed gene activities, grouping genes into functional sets can often provide deeper insights into the underlying biological system. Indeed, the activity of certain pathways, particularly those that regulate cellular metabolism, has been found to be a strong predictor of complex phenotypes and response to treatment, both at the level of cells as well as that of individual patients \cite{pathwayTherapy1, pathwayTherapy0}.

A {\it biological pathway} can be defined as a collection of molecules that coordinate to perform a specific action or change in the cell. This change could involve production of a new molecule, movement, growth or a physical transformation, or even cell death. While the activity of a particular pathway can be understood qualitatively based on the phenotypic changes in a cell, its quantitative estimation relies on the relative proportion of RNA molecules produced. The use of gene expression as a proxy for activity rests on the notion that the amount of mRNA molecules produced represent the economic resources of the cell \cite{CostOfGene}.  Despite an understanding of how individual molecules in a biological pathway orchestrate a particular cellular function, it remains unclear whether the distribution of certain gene modules, and by proxy the corresponding pathway activities, are shared across cell types \cite{metabcons}. It has recently been shown that vastly different cell types contain similar ratios of metabolic enzymes, and tightly control the amounts of specialized proteins produced \cite{enzymeCell}, highlighting previously unappreciated similarities among cell types. Nonetheless, the extent to which the distribution of relative mRNA molecules for genes in a given pathway might be consistent across different cell types remains elusive. Are there certain pathways which maintain a similar activity across cell types? This question is of fundamental significance because if true, it suggests a widespread design principle of cell biology.

With the advent of single cell RNA sequencing, it is now possible to study distinct but closely related cell populations \cite{RNAreview1} and examine the aforementioned question. Single cell RNA-sequencing (scRNA-seq) allows us to measure gene expression information from tens of thousands of individual cells, unraveling the cellular heterogeneity of a tissue in unprecedented detail. The resulting data can be thought of as a $c \times g$ matrix, $\bm \eta=((\eta_{ab}))_{1 \leq a \leq c, 1 \leq  b \leq g}$, where $c$ corresponds to the number of cells, and $g$ refers to the number of genes, and each entry $\eta_{ab}$ corresponds to the number of RNA molecules detected for a given gene $a$ in some cell $b$. The high-dimensional, multisample (corresponding to multiple cell-types) nature of a typical scRNA-seq experiment makes this a fitting application of the multi-sample crossmatch test.

\subsection{Data Overview and Study Setup}

We apply our method on scRNA-seq data generated from purified T cell populations found in three tissue locations: (a) peripheral blood (hereafter referred to as \textsf{blood}), (b) tumor-infiltrating immune cells (hereafter referred to as \textsf{tumor}), and (c) normal tissue adjacent to the tumor from the same organ (hereafter referred to as \textsf{adj. normal}). We examined the following two datasets, where T cells extracted from each location were assigned a particular subtype based on flow cytometry and expression of known canonical cell surface proteins.\footnote{Both the datasets used are open access, and available in the Gene Expression Omnibus (GEO). The raw sequencing data for T cells for the Hepatocellular Carcinoma dataset can be obtained from the GEO entry GSE98638 and the European Genome-phenome Archive database entry EGAS00001002072. The single cell sequencing data corresponding to the Non-Small-Cell Lung cancer case study can be found at GSE99254 and EGAS00001002430.} 

\begin{itemize}

\item[(1)] \textsf{Non-Small-Cell Lung Cancer} (NSCLC \cite{LungSC}) dataset: Here, the T cell subtypes found at each location were: $CD8^{+}$ Cytotoxic, $CD4^{+}$ Naive, $CD4^{+}$ Regulatory T cells (denoted by $T_{\mathrm{reg}}$), and $CD4^{+}$ Naive Helper.

\item[(2)] \textsf{Hepatocellular Carcinoma} (HCC \cite{liversinglecell}) dataset: Here, the T cell subtypes profiled were: $CD8^{+}$ Cytotoxic, $CD4^{+}$ Naive, and $CD4^{+}$ Regulatory T cells ($T_{\mathrm{reg}}$). 

\end{itemize}

Recall that single cell data is extremely sparse count data with complex correlation structure. Nonetheless, the two datasets are comparable in terms of the sequencing protocol used, data generation, and the technical quality of the data (Table \ref{tab:tabdata}). In both the datasets we used deep sequencing was performed, ensuring our ability to detect genes with low expression. The summary of the number of cells sequenced for each cell type is provided in Table \ref{tab:Tcelldist}. 

\begin{table}[h]
	\begin{center}
		\begin{tabular}{rrr}
			& \textsf{NSCLC} & \textsf{HCC} \\ \hline 
			Mean reads per cell & 1,040,000 & 1,100,000 \\
			Median genes per cell & 2,859 & 2,702 \\
			Total number of T cells profiled & 12,210 & 4,794 \\ 
			\hline
		\end{tabular}
	\caption{\small{Summary characteristics of the two scRNA-seq datasets: \textsf{Non-Small-Cell Lung Cancer} (NSCLC) and \textsf{Hepatocellular Carcinoma} (HCC).} For scRNA-seq it has been shown that with half a million reads per cell, most genes expressed can be detected, and that one million reads are sufficient to estimate the mean and variance of gene expression \cite{SeqDepth}.}
\label{tab:tabdata} 
	\end{center}
	\vspace{-0.15in}	
\end{table}

\normalsize

\begin{table}[h]
\small{
	\begin{center} 
		\begin{tabular}{rrrrr}
		\textsf{Non-small-cell Lung Cancer} & ($K=4$ groups)  &   &  \\
		\hline
			Tissue Type & $CD8^{+}$ Cytotoxic & $CD4^{+}$ Naive & $CD4^{+} \; T_{\mathrm{reg}}$ & $CD4^{+}$ Helper \\\hline
			\textsf{Adj. Normal} (2115 cells)  & 934  & 655  &  288  & 238 \\
			\textsf{Tumor} (5835 cells) & 2182  &  1591  &  1170  & 892 \\
			\textsf{Blood} (4260 cells) & 1323 & 1254  & 1011 & 672 \\
			\hline
			&  &   &    & \\
			\textsf{Hepatocellular Carcinoma} & ($K=3$ groups)  &   &  \\
			\hline
			Tissue Type & $CD8^{+}$ Cytotoxic &  $CD4^{+}$ Naive & $CD4^{+} \; T_{\mathrm{reg}}$ \\\hline
			\textsf{Adj. Normal} (997 cells) &  412 &  406  & 179 \\
			\textsf{Tumor} (2170 cells) & 563  & 515  & 549  \\
			\textsf{Blood} (1627 cells) & 777 & 606 & 787 \\ 
			\hline 
		\end{tabular}
		\caption{\small{The number of cells sequenced for various T cell subtypes in each of the two cancer settings. These correspond to the sample sizes of the different groups in the $K$-sample hypothesis testing problem. (\textsf{Adj. Normal} denotes tumor-adjacent normal tissue from the same organ, $T_{\mathrm{reg}}$ stands for regulatory T cell, and CD abbreviates \emph{cluster of differentiation} or \emph{classification determinant}, a protocol used for immunophenotyping cells.) }}	\vspace{-0.15in}	
\label{tab:Tcelldist}
	\end{center}
}	
\end{table}

\normalsize

The metabolic state of T cells is implicated in diseases such as cancer, wherein the tumor microenvironment enforces dysfunctional T cell metabolism, thereby negatively affecting their anticancer functionality \cite{TcellMetab}. Thus, for each of the two cancer datasets and in the three tissue locations described above (namely, \textsf{Adj. Normal}, \textsf{Tumor}, and \textsf{Blood}), we examine whether the distribution of gene sets that correspond to biological metabolic pathways are consistent across T cell subtypes (therefore, the number of different T cell subtypes corresponds to the number of classes $K$). We specifically focus on 86 metabolic pathways described in the Kyoto Encyclopedia of Genes and Genomes (KEGG \cite{KEGG}). We obtained a list of genes corresponding to each pathway, and subsequently ascribed genes into 86 subsets, each subset corresponding to a metabolic pathway. Then, for each one of the 3 different tissue locations and for each one of the 86 pathways, we tested the null hypothesis (using the MCMM test \eqref{eq:SKN} described above) that the multivariate distribution of the genes belonging to that particular pathway is alike across the $K$ different T cell subtypes (recall that $K=4$ in the \textsf{Non-Small-Cell Lung Cancer} dataset, and $K=3$ in the \textsf{Hepatocellular Carcinoma}  dataset). In each of the cases the corresponding sample sizes for the $K$ groups are given in the rows of Table \ref{tab:Tcelldist}. The number of genes in a given metabolic pathway ranged approximately between 30 and 110. To account for multiple-hypothesis testing, the resulting $p$-values are adjusted using the Benjamini-Hochberg (BH) correction procedure \cite{BHCorrection}.

This study design allowed us to understand which metabolic pathways change in distribution across distinct, but closely-related, T cell subtypes. In particular, assessing the distribution of gene sets that belong to a particular metabolic pathway across the T cell subtypes allows us to address the following questions: (1) Which pathways have a similar distribution across T cell subtypes in a given tissue? (2) Are there pathways that have a stable and comparable distribution across T cell subtypes in a normal/healthy tissue, but a heterogenous or perturbed distribution in a tumor? (3) For pathways that have a disparate distribution across the T cell subtypes, which subtypes show the most distinct distribution?

\subsection{Comparing Pathway Distributions Based on Tissue Location}
\label{sec:dataresults_I}

The following is the outcome of the BH-corrected MCMM tests for assesing metabolic pathway distributions across the different T cell subtypes in the two datasets: 

\begin{itemize}

\item In the NSCLC dataset, we found that of the 86 pathways examined, our test did not reject the null hypothesis for merely 35 pathways in the \textsf{Tumor} tissue, compared to 56 and 74 pathways in the \textsf{Blood} and \textsf{Adj. Normal} tissues, respectively.

\item In HCC dataset, the set of pathways for which we failed to reject the null were remarkably alike to that for NSCLC.  Specifically, 41, 63 and 76 metabolic pathways were undistinguishable across T cell subtypes in the \textsf{Tumor}, \text{Blood} and \textsf{Adj. Normal} tissues, respectively.  
\end{itemize} 

The fact that majority of the metabolic pathways do not show evidence for dissimilar distribution across cell types based on our test indicates that the T cell subtypes might be more similar than previously appreciated in terms of how they regulate their basic metabolic machinery.  Further, for each pair of tissue location, we computed the overlap in the pathways for which the null hypothesis was accepted or rejected. As expected, we found that the concordance between the results was substantially higher for \textsf{Blood} and \textsf{Adj. Normal}, than either of these tissues had with the \textsf{Tumor} tissue (this is seen from the off-diagonal values in the tables in Figure  \ref{fig:twobytwo}).

\begin{figure}[h]
	\begin{center}
		\includegraphics[width=6.1in]{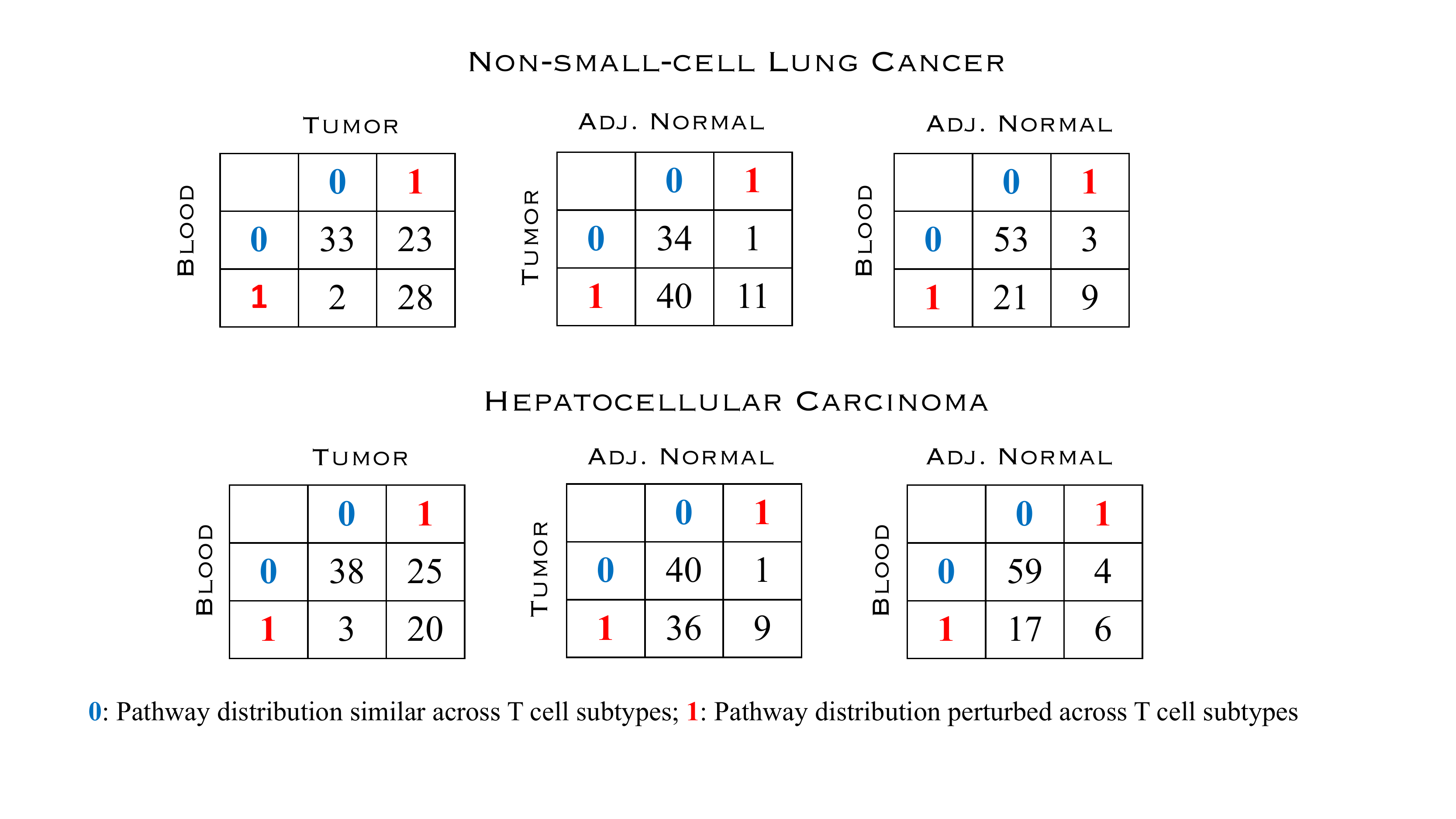}\vspace{-0.35in}
	\end{center} 
\caption{\small{A set of $2 \times 2$ tables showing how many of the 86 hypothesis (corresponding to the pathways) were accepted/rejected for each of 3 pairs of tissue locations. Here, $\blue{\textbf{0}}$ stands for pathways whose distribution are either similar/stable (null hypothesis accepted), and $\red{\textbf{1}}$ stands for pathways whose distributions are perturbed/heterogenous (null hypothesis rejected), across the T cell subtypes in two types of solid organ malignancies.}} \label{fig:twobytwo}
\end{figure}
\normalsize

Interestingly, we found that 8 metabolic pathways were differentially distributed across the T cell subtypes in each tissue examined in NSCLC whereas 5 pathways exhibited this pattern in the HCC dataset (Figure \ref{fig:VennDiag}). The {\it purine metabolism} pathway was a common pathway shared by both datasets which showed evidence for heterogenous distribution among the T cell subtypes in every tissue. This discovery suggests that purine metabolism is fundamentally different even among the closely related T cell subtypes. In order to find out which T cell subtype contributed most to this difference for the metabolic pathways that emerged as being heterogeneously distributed in all tissues, we employed a class selection procedure described below. To select a single class, we looked at all the pairwise comparisons where the null was rejected, and then identified the class that was common across all the cases of rejection (Figure \ref{fig:VennDiag}c-d).

We found that $CD4^{+}$ regulatory T cells ($T_{\mathrm{regs}}$) were the strongest contributors as to why a pathway, such as purine metabolism, was detected as being differentially distributed. It is important to note that, purine metabolism regulates the balance of proinflammatory and immunosuppressive molecules produced by T cells, and the activation of the purinergic receptor P2X7 has been shown to inhibit the immunosuppressive functions specifically in $T_{\mathrm{regs}}$ \cite{PurineTReg}. Thus, the fact that our test discovered purine metabolism as being differentially distributed across T cell subtypes in both studies, and specifically in $T_\mathrm{{regs}}$, showcases its ability to unearth meaningful biological phenomena.

\begin{figure}[h]
	\begin{center}
		\includegraphics[width=6.1in]{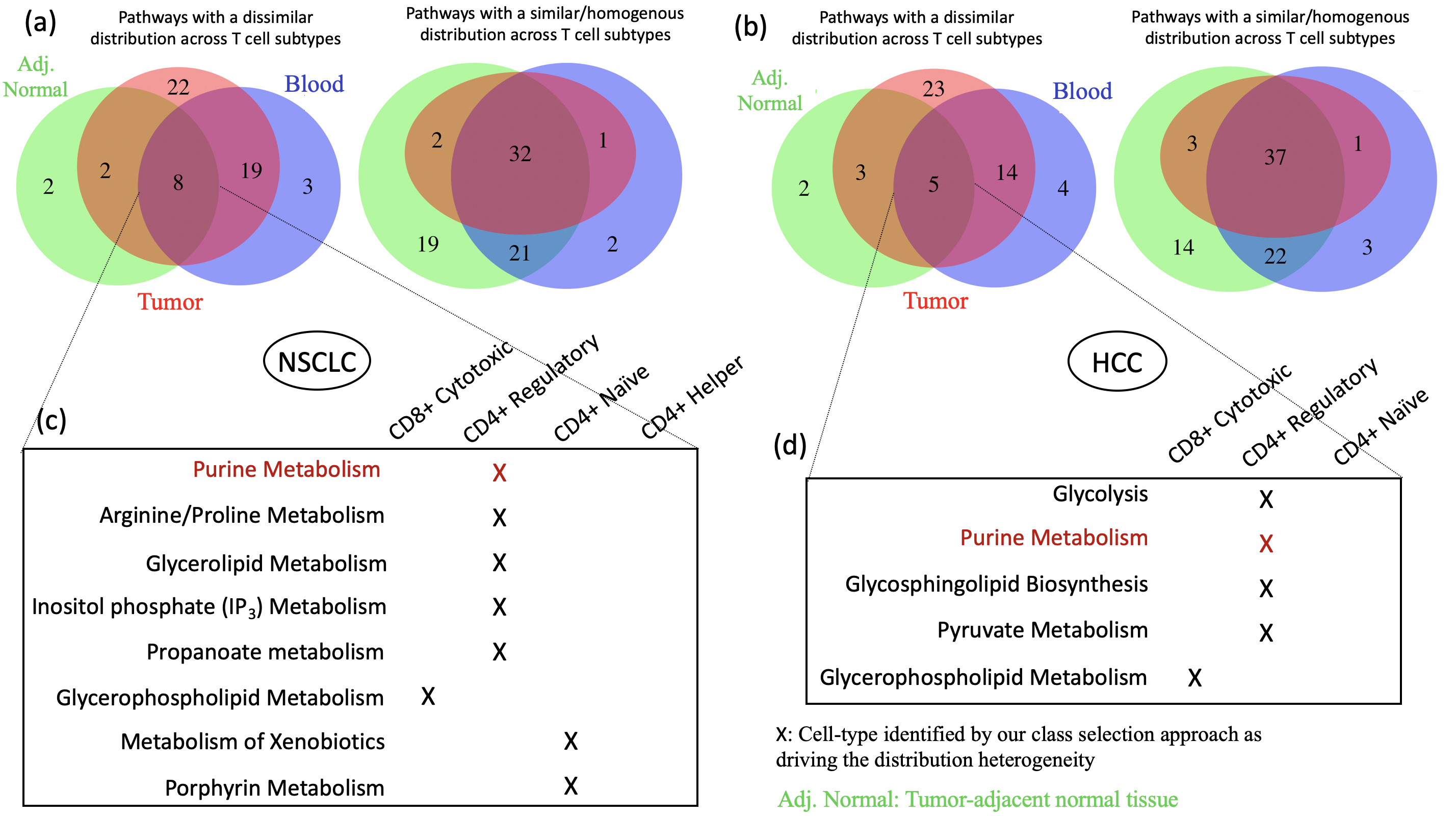}
	\end{center}
	\caption{\small{(a), (b) Venn diagrams showing the overlap among the metabolic pathways that demonstrated evidence for differential and similar distributions across the T cell subtypes in NSCLC and HCC, respectively.  For the pathways that were heterogenous among the T cell classes across all 3 tissue types, we used the \emph{class selection} procedure to identify which T cell subtype had the most disparate distribution (c), (d). Differential distribution for the purine metabolism pathway, driven largely by $T_{\mathrm{regs}}$, was observed in both NSCLC and HCC. }}
\label{fig:VennDiag}	
\end{figure}

\subsection{Differentially Distributed Pathways as Biological and Algorithmic Features}
\label{sec:dataresults_II}

In the NSCLC dataset, we found that 51 pathways showed evidence for differential distribution across the tumor-infiltrating T cell subtypes, whereas 45 pathways exhibited this pattern in the tumor-infiltrating T cells in the HCC dataset. We observed that if a metabolic pathway demonstrated an intra-T cell-type heterogenous distribution in \textsf{Blood} and \textsf{Adj. Normal}, that heterogeneity was preserved in the \textsf{Tumor} tissue (in the first and third Venn diagrams in Figure \ref{fig:VennDiag}, the intersection of the green (\textsf{Adj. Normal}) and the blue circles (\textsf{Blood}) is completely contained in the orange circle corresponding to \textsf{Tumor}). In other words, rejecting the null for a particular pathway for T cells in the \textsf{Adj. Normal} and \textsf{Blood} tissues was useful in prognosticating that pathway's behavior in the tumor-infiltrating T cells.  On the other hand, we found that certain pathways that demonstrated a dissimilar distribution across T cell subtypes in the \textsf{Tumor} and \textsf{Adj. Normal} tissues showed evidence for homogenous distribution in blood (the intersection of the green and orange circles minus the blue circle in the first and third Venn diagrams in Figure \ref{fig:VennDiag}). Specifically, this pattern was true for 2 pathways (sphingolipid metabolism and glycophospholipid synthesis) in the NSCLC dataset, and for 3 pathways (phenylalanine metabolism, oxidative phosphorylation and O-glycan biosynthesis) in the HCC dataset.  This difference can be attributed, at least partially, to the organ-specific function of these pathways. For instance, ceramide, a central molecule in the sphingolipid metabolism, regulates endothelial permeability and airway smooth muscle function in the lungs \cite{CeramideLung}, and increased sphingolipid metabolism is a hallmark of lung cancer \cite{sphingLungCancer}. Similarly, phenylalanine hydroxylase, the main enzyme in the phenylalanine metabolism pathway is active exclusively in the liver \cite{PKULiver}, and otherwise inactive in the blood. Hence, in light of the  organ-specific roles of these pathways, it is rather reassuring that our test appropriately rejects the null hypothesis and captures their heterogenous distribution among the T cell subtypes based on the tissue location.

Clustering and cell type identification are crucial steps in single cell data analysis. scRNA-seq analysis pipelines often first identify highly variable genes in the dataset, and subsequently use those genes as an input to the t-distributed stochastic neighbor embedding (tSNE) algorithm. Even with this approach, however, single cell analysis pipelines typically fail to resolve the different T cell populations, and rarely identify the different immune cell subtypes in any reliable or discrete fashion. We wondered whether in practice, one might be able to utilize our test to identify differentially distributed pathways and then use this information to improve the tSNE-based clustering approach. To investigate whether the pathways that emerge as being differentially distributed can serve as meaningful features in identifying cell types, we focussed on the pathways that were consistently identified as being differentially distributed across the T cells in all three tissue types in NSCLC. To our surprise, we found that using the subset of genes corresponding to the pathways that our analysis identified as being heterogeneously distributed improves the clustering results (Figure \ref{fig:tSNE}(a)). Moreover, the clustering results converge with, and act as an indirect validation for our class selection approach because the latter identified $CD4^{+}$ $T_{\mathrm{regs}}$ as the cell type with the dissimilar distribution for metabolic pathways such as $IP_3$ metabolism, purine metabolism, and arginine and proline metabolism. When we used the genes comprising these pathways as the input to tSNE (Fig. \ref{fig:tSNE}a), we observed that indeed the $CD4^{+}$ $T_{\mathrm{regs}}$ became more easily visually distinguishable compared to using the genes in other pathways which had a similar distribution across the T cell subtypes (Figure \ref{fig:tSNE}(b)). 

\begin{figure}
	\begin{center}
		\includegraphics[width=6.1in]{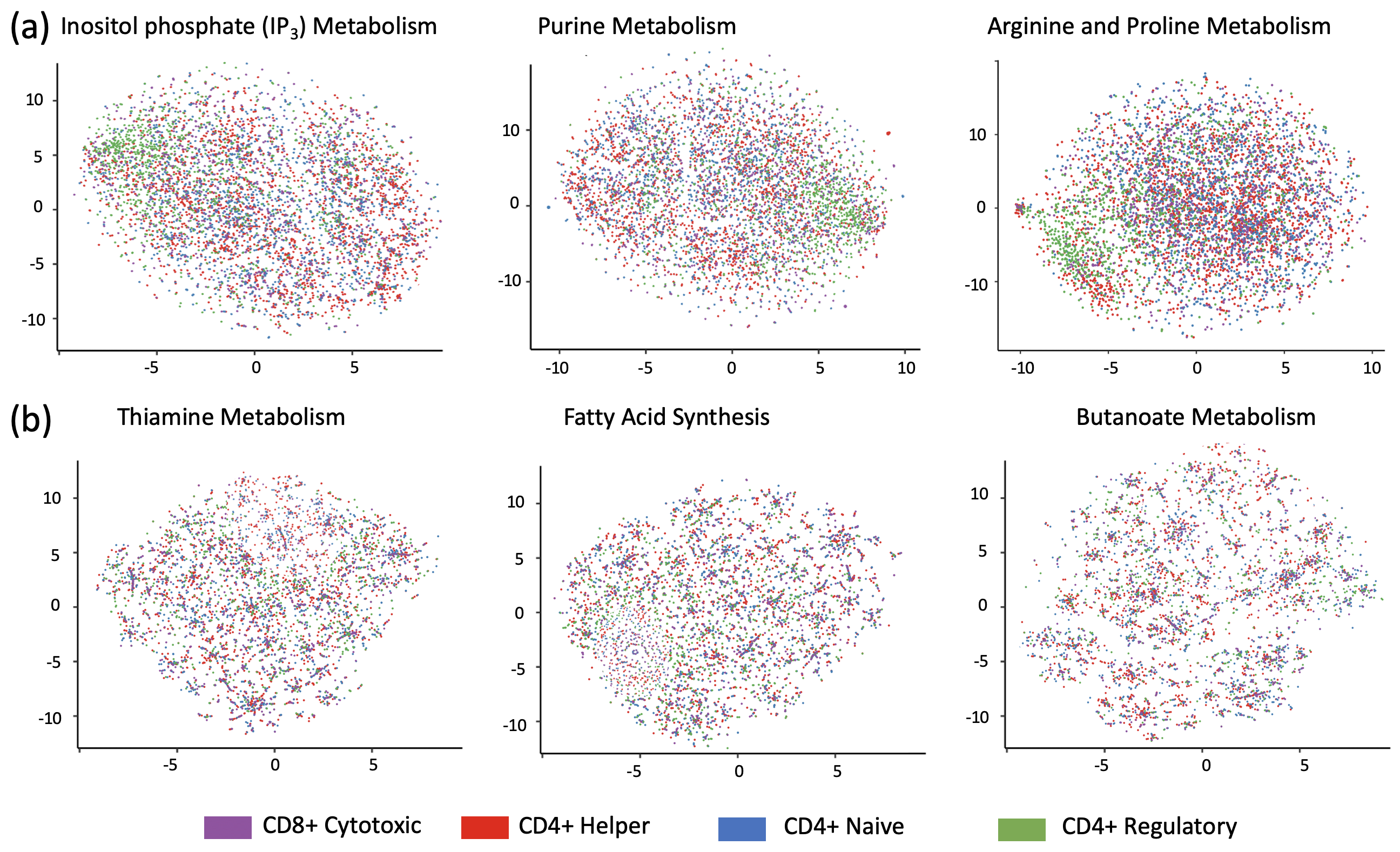}
	\end{center}
	\caption{\small{Results of the tSNE analysis performed using a specific pathway to cluster the 4 different tumor-infiltrating T cell subtypes in NSCLC: (a) pathways that were identified as differentially distributed successfully manage to segregate $CD4^{+}$ $T_{\mathrm{regs}}$ as predicted by our class selection procedure, whereas (b) pathways that were identified as being similar across the T cell subtypes produce a more patchy clustering wherein the subtypes are visually indistinguishable. All tSNE plots were generated with \emph{perplexity} = 30.}} \label{fig:tSNE}\vspace{-0.1in}
\end{figure}

\section{Discussion}

This paper introduces a novel graph-based nonparametric test for comparing multiple multivariate distributions. Using optimal matching as the basis, we demonstrate that our test, in addition to other multisample generalizations of Rosenbaum's crossmatch test \cite{rosen}, is distribution-free, computationally efficient, and consistent for general alternatives, making it particularly attractive for modern high-dimensional statistical applications. We also obtain a  joint central limit theorem for the entire matrix of cross-counts, and, hence, derive a distributional limit theorem for the test statistics, under general alternatives. Our numerical experiments demonstrate that the proposed method outperforms other non-parametric graph-based multisample methods as well as commonly used parametric tests, in a variety of simulation settings. Lastly, we showcase the utility of this test in the field of single cell transcriptomics where we used our test to address an important question in signal transduction and cell biology. Our multisample procedure uncovered revealing patterns about how closely-related, key immune cells in the body (namely, T cell subtypes), might alter their metabolic machinery in solid organ malignancies, particularly depending on the tissue location. We envision that this test and the underlying theoretical intuitions described in this work will find broad applicability in future research that examines hypothesis testing in the multisample, multivariate framework. Furthermore, our test opens up new paradigms of investigation for practitioners to assess its properties and its feasibility in being adapted to other algorithms, as we have demonstrated here through its usage as a pre-processing step in tSNE-based data visualization.

\appendix

\section{Proofs from Section \ref{sec1}}\label{app1}

In this section we present the proofs of the results from Section \ref{sec1}. The proof of Proposition \ref{ppn:mean_var}, which computes the mean and the variance of $\vec{\bm A}_N$ is given in Section \ref{sec:pfmean_var}. The proof of the asymptotic null distribution in Theorem \ref{nullasmS} is described in Section \ref{sec:pfdistribution_A}, and the proof of the consistency (Theorem \ref{consmain}) is in Section \ref{app2}.

\subsection{Proof of Proposition \ref{ppn:mean_var}}
\label{sec:pfmean_var}
Since the conditional distribution of $\bm A_N$ given $\mathscr{X}$ equals the unconditonal distribution of $\bm A_N$ under the null, we can assume that the edges of $\sG(\mathscr{X})$ are fixed, say $\{(1,2), (3,4),~\ldots,(N-1,N)\}$ without loss of generality. For notational convenience, let us define $\bm \eta := (\eta_1,\ldots,\eta_K)$ and $\bm N := (N_1,\ldots,N_K)$. Following the notations in subsection \ref{sec:pfdistribution_A}, it follows that:
\begin{equation}\label{exproof}
\mathbb{E}_{H_0}(a_{st}) = \sum_{j=1}^I \mathbb{P}_{H_0}\left(\{L_{2j-1}, L_{2j}\} = \{s,t\}\big|\bm \eta = \bm N\right).
\end{equation}
By an easy sampling without replacement argument, for each $j \leq I$, 
\[\mathbb{P}_{H_0}\left(\{L_{2j-1}, L_{2j}\} = \{s,t\}\big| \bm \eta = \bm N \right) =   \left\{
\begin{array}{ll}
\frac{2N_sN_t}{N(N-1)}& \textrm{if}~s< t, \\ \\ 
\frac{N_s(N_s-1)}{N(N-1)} & \textrm{if}~s= t.
\end{array} 
\right. \] 
The result in \eqref{eq:expectationH0} now follows from \eqref{exproof} on observing that $I = N/2$. 

Next, note that for $1\leq s_1 \neq s_2 \leq K$,
\begin{align*}
\mathrm{Var}_{H_0} & (a_{s_1s_2}) = \sum_{j=1}^I \mathrm{Var}_{H_0} \left(\mathbf{1}\left\{\{L_{2j-1}, L_{2j}\} = \{s_1,s_2\}\right\}\big|\bm \eta = \bm N\right)\\
&+\sum_{j_1\neq j_2} \mathrm{Cov}_{H_0}\left(\mathbf{1}\left\{\{L_{2j_1-1}, L_{2j_1}\} = \{s_1,s_2\}\right\}, \mathbf{1}\left\{\{L_{2j_2-1}, L_{2j_2}\} = \{s_1,s_2\}\right\}\big|\bm \eta = \bm N\right).
\end{align*}
First, note that
\begin{align}\label{ce1}
\sum_{j=1}^I \mathrm{Var}_{H_0} \left(\mathbf{1}\left\{\{L_{2j-1}, L_{2j}\} = \{s_1,s_2\}\right\}\big|\bm \eta = \bm N\right) &= I\left[\frac{2N_{s_1}N_{s_2}}{N(N-1)}\right]\left[1-\frac{2N_{s_1}N_{s_2}}{N(N-1)}\right]\nonumber\\
&=\frac{N_{s_1}N_{s_2}}{N-1}\left[1-\frac{2N_{s_1}N_{s_2}}{N(N-1)}\right]. 
\end{align}
Next, we have 
\begin{align}\label{ce2}
& \sum_{j_1\neq j_2} \mathrm{Cov}_{H_0} \left(\mathbf{1}\left\{\{L_{2j_1-1}, L_{2j_1}\} = \{s_1,s_2\}\right\}, \mathbf{1}\left\{\{L_{2j_2-1}, L_{2j_2}\} = \{s_1,s_2\}\right\}\big|\bm \eta = \bm N\right)\nonumber\\
&= \sum_{j_1\neq j_2} \mathbb{P}_{H_0} \left(\{L_{2j_1-1}, L_{2j_1}\} = \{L_{2j_2-1}, L_{2j_2}\} = \{s_1,s_2\}\big|\bm \eta = \bm N\right) - \sum_{j_1 \neq j_2}\left(\frac{2N_{s_1}N_{s_2}}{N(N-1)}\right)^2\nonumber\\
&= I(I-1)\frac{4N_{s_1}N_{s_2} (N_{s_1}-1) (N_{s_2 }-1)}{N(N-1)(N-2)(N-3)} - I(I-1)\left(\frac{2N_{s_1}N_{s_2}}{N(N-1)}\right)^2\nonumber\\
&= \frac{N_{s_1}N_{s_2} (N_{s_1}-1) (N_{s_2 }-1)}{(N-1)(N-3)} - \frac{N_{s_1}N_{s_2}}{N-1}\left[\frac{N_{s_1}N_{s_2}}{N-1}-\frac{2N_{s_1}N_{s_2}}{N(N-1)}\right]
\end{align}
Adding \eqref{ce1} and \eqref{ce2}  gives the expression for $\mathrm{Var}_{H_0} (a_{s_1 s_2})$ given in Proposition \ref{ppn:mean_var}. Next, take $1\leq s_1 \neq s_2 \neq s_3 \leq K$. Then, 
\begin{align}\label{ce33}
\mathbb{E}_{H_0} \left(a_{s_1s_2}a_{s_1s_3}\right) &= \sum_{j_1\neq j_2}\mathbb{P}_{H_0} \left(\{L_{2j_1-1}, L_{2j_1}\} = \{s_1,s_2\},\{L_{2j_2-1}, L_{2j_2}\} = \{s_1,s_3\}\big|\bm \eta = \bm N\right)\nonumber\\
&= I(I-1) \frac{4N_{s_1}(N_{s_1}-1)N_{s_2}N_{s_3}}{N(N-1)(N-2)(N-3)}\nonumber\\
&= \frac{N_{s_1}(N_{s_1}-1)N_{s_2}N_{s_3}}{(N-1)(N-3)}.
\end{align}
	The expression for $\mathrm{Cov}_{H_0}(a_{s_1s_2},s_{s_1,s_3})$ now follows from \eqref{ce33} on observing that:
	$$\left(\mathbb{E}_{H_0} a_{s_1s_2}\right)\left(\mathbb{E}_{H_0} a_{s_1 s_3}\right) = \frac{N_{s_1}^2 N_{s_2}N_{s_3}}{(N-1)^2}.$$ Finally, take $1\leq s_1\neq s_2\neq s_3\neq s_4\leq K$. In this case, 
\begin{align}\label{ce44}
\mathbb{E}_{H_0} \left(a_{s_1s_2}a_{s_3s_4}\right)&=\sum_{j_1\neq j_2}\mathbb{P}_{H_0} \left(\{L_{2j_1-1}, L_{2j_1}\} = \{s_1,s_2\},\{L_{2j_2-1}, L_{2j_2}\} = \{s_3,s_4\}\big|\bm \eta = \bm N\right)\nonumber\\
&= I(I-1) \frac{4N_{s_1}N_{s_2}N_{s_3}N_{s_4}}{N(N-1)(N-2)(N-3)}\nonumber\\
&= \frac{N_{s_1}N_{s_2}N_{s_3}N_{s_4}}{(N-1)(N-3)}.
\end{align}
	The expression for $\mathrm{Cov}_{H_0}(a_{s_1s_2},s_{s_3,s_4})$ now follows from \eqref{ce44} on observing that:
	$$\left(\mathbb{E}_{H_0} a_{s_1s_2}\right)\left(\mathbb{E}_{H_0} a_{s_3 s_4}\right) = \frac{N_{s_1} N_{s_2}N_{s_3}N_{s_4}}{(N-1)^2}.$$
	This completes the proof of Proposition \ref{ppn:mean_var}. \qed
\subsection{Proof of Theorem \ref{nullasmS}}
\label{sec:pfdistribution_A}

Note that it suffices to prove \eqref{ant}. The remaining assertions in Theorem \ref{nullasmS} is an immediate consequence of \eqref{ant}.  The proof of \eqref{ant} proceeds in two-steps: (1) $\frac{1}{\sqrt {N}} (\vec{\bm A}_N- \E_{H_0}(\vec{\bm A}_N)) \dto N(0, \bm \Gamma)$, for some non-negative definite matrix $\bm \Gamma$, and (2) $\frac{1}{N} \Cov_{H_0} (\vec{\bm A}_N) \rightarrow \bm \Gamma$, and $\bm \Gamma$ is invertible. 

We begin with the proof of (1): Denote the pooled sample 
$$\sX=(X_1^{(1)},\ldots, X_{N_1}^{(1)},\ldots, X_{1}^{(K)},\ldots, X_{N_K}^{(K)})$$ (forgetting the labels) as $\cZ_N:=(Z _1, Z_2, \ldots, Z_N)$. Under the null $H_0$, $Z _1, Z_2, \ldots, Z_N$ are i.i.d. $F$ (the unknown null distribution). Let $L_1,\ldots,L_N$ be i.i.d. random variables, taking value in $\{1, 2, \ldots, K\}$, independent of $Z _1,\ldots, Z_N$, such that 
\begin{align}\label{eq:L_prob}
\mathbb{P}(L_1 = s) = \frac{N_s}{N}, \quad \textrm{for all}~ s\in [K].
\end{align}
For each $s \in [K]$, define $\eta_s = \sum_{i=1}^N \bm 1\{L_i = s\} \sim \dBin(N, \frac{N_s}{N})$ and for each $x,y \in \mathbb{R}^d$, let $e(x,y) := \bm 1 \left\{(x,y) \in  E(\sg(\cZ_N \cup \{x, y\}))\right\}$.  Define the $K\times K$ matrix $\bm B_N=(b_{st})_{1 \leq s, t \leq K}$ as follows:
\[b_{st} =   \left\{
\begin{array}{ll}
\sum_{1\leq i\neq j\leq N}e(Z_i,Z_j) \bm 1 \{L_i=s, L_j = t\} & \textrm{if}~s\neq t, \\ \\ 
\frac{1}{2}\sum_{1\leq i\neq j\leq N} e(Z_i, Z_j) \bm 1 \{L_i=L_j = s\} & \textrm{if}~s= t.
\end{array} 
\right. \]
Under $H_0$, it follows from Lemma \ref{condrel} that the conditional distribution of $\bm B_N$ given $(\eta_1,\ldots,\eta_{K}) = (N_1,\ldots,N_{K})$ is same as the distribution of $\bm A_N$. Therefore, it suffices to derive the limiting distribution of $\bm B_N|\{(\eta_1,\ldots,\eta_{K}) = (N_1,\ldots,N_{K})\}$.

To this end, note that conditional on $\cZ_N$, the matching graph $\sG(Z_N)$ is fixed and  since the $I=N/2$ matched edges in the graph are disjoint, the samples in $\cZ_N$ can be (re)-labelled $1, 2, \ldots, N$ such that $\{(1, 2), (3, 4), \ldots, (N-1, N)\}$ are the $I$ matched edges. Then the elements of the matrix $\bm B_N$ can be written as 
$$b_{st}\stackrel{D}=\sum_{j=1}^I \left(\bm 1\{L_{2j-1}=s, L_{2j}=t\} + \bm 1\{L_{2j-1}=t, L_{2j}=s\}\right).$$
Moreover, $\eta_s= \sum_{j=1}^I \bm 1\{L_{2j-1}=s\}+\sum_{j=1}^I \bm 1\{L_{2j}=s\}$. Therefore, conditional on $\cZ_N$ the ${K \choose 2}+K$ vector $\bm V_N:=(\vec{\bm B}_N,\eta_1,\ldots,\eta_K)'$ can be written as the sum of $I$ i.i.d. random vectors. This implies, under $H_0$, as $N \rightarrow \infty$, by the multivariate CLT, 
\begin{align}\label{eq:VN}
\frac{\bm V_N- \E_{H_0}(\bm V_N)}{\sqrt I} \Big|\cZ_N \dto N_{{K \choose 2}+K}(0, \bm \Gamma_0),
\end{align}
where 
$$\bm \Gamma_0 := \mathrm{Cov}(\vec{\bm B}, \bar{\eta}_1, \bar{\eta}_2, \ldots, \bar{\eta}_K),$$
where 
\begin{itemize}
\item[--] $\bm B=((\bm 1\{\bar L_1=s, \bar L_2=t\}))_{1\leq s, t \leq K}$, and 
\item[--] $\bar \eta_s=\bm 1 \{\bar L_1=s\}+\bm 1 \{\bar L_2=s\}$, and 
\item[--] $\bar{L}_1, \bar{L}_2$ are i.i.d. random variables taking value $s$ with probability $p_s$, for $1 \leq s \leq K$. (This is the limit of the random variable $L_1$ defined in \eqref{eq:L_prob}.)
\end{itemize}
As the RHS in \eqref{eq:VN} does not depend on the conditioning event, the unconditional limit is also the same:
$$\frac{\bm V_N- \E_{H_0}(\bm V_N)}{\sqrt I}  \dto N_{{K \choose 2}+K}(0, \bm  \Gamma_0).$$
Then by \cite[Theorem 2]{holst}, there exists a ${K \choose 2} \times {K \choose 2}$ matrix $\bm \Gamma$ such that, under $H_0$,  
\begin{align}\label{eq:limit_sqrtn}
\frac{1}{\sqrt {N}}\left(\vec{\bm A}_N- \E_{H_0}(\vec{\bm A}_N)\right) &\stackrel{D}=
\frac{1}{\sqrt N}\left(\bm {\vec B}_N-\E_{H_0}(\bm {\vec B}_N) \right)|\{(\eta_1,\ldots,\eta_{K}) = (N_1,\ldots,N_{K})\} \nonumber \\ 
& \dto N_{{K \choose 2}}(0, \bm  \Gamma),
\end{align} 
which completes the proof of (1). 

To see (2) note that the fourth moments of the elements of $\frac{\vec{\bm A}_N- \E_{H_0}(\vec{\bm A}_N)}{\sqrt N}$ are bounded. Therefore, by uniform integrability, $\frac{1}{N}\Cov_{H_0}(\vec{\bm A}_N) \rightarrow \bm \Gamma$. The invertibility of $\Cov_{H_0}(\vec{\bm A}_N)$ (and hence $\bm \Gamma$) follows from Lemma \ref{invert1} (in Appendix \ref{sec:technical_lemmas}). The result in \eqref{ant} then follows from \eqref{eq:limit_sqrtn} and an application of the Slutsky's theorem.

\subsection{Proof of Theorem \ref{consmain}}\label{app2}

The entry-wise almost sure limit of $\frac{1}{N} \bm A_N$ as in \eqref{eq:HPab} is a direct consequence of \cite[Proposition 1]{castropell} (by choosing $\phi = \sum_{s=1}^K p_s f_s$ and $\phi_N = \frac{1}{N}\sum_{s=1}^K N_s f_s$ in \cite[Proposition 1]{castropell}).

Now, to prove consistency we need show that the test statistics have different limits under the null and the alternative. To this end, we have the following lemma. 

\begin{lem}\label{componentwise} Let $H(f_1,f_2,\ldots,f_K)$ be as defined in \eqref{eq:R_consistency}. Then 
$$H(f_1,f_2,\ldots,f_K) \leq H(f,f,\ldots,f),$$ and equality holds if and only if $f_1=f_2=\cdots=f_K$ outside a set of Lebesgue measure $0$.
\end{lem}

\begin{proof} It follows from the Cauchy-Schwarz inequality, that for every $s \in [K]$,
\begin{equation}\label{intm1}
\int_{\mathbb{R}^d} \frac{p_s^2 \left(f_s(z)\right)^2}{\sum_{u=1}^K p_u f_u(z)}\mathrm dz \geq \frac{\left(\int_{\mathbb{R}^d} p_s f_s(z)\mathrm dz\right)^2}{\int_{\mathbb{R}^d} \sum_{u=1}^K p_u f_u(z)\mathrm dz} = p_s^2. 
\end{equation}
This implies, $H(f_1,f_2,\ldots,f_K)  = \frac{1}{2}-\mathrm{tr}(\bm H) \leq \frac{1}{2}-\frac{1}{2}\sum_{s=1}^K p_s^2 = H(f,f,\ldots,f)$, as required. 
	
Now, note that equality holds in \eqref{intm1} if and only if $p_s f_s = c_s \sum_{u=1}^K p_u f_u$, for some constant $c_s$, almost everywhere. Integrating both sides of the last relation, gives  $c_s = p_s$, that is, $f_s = \sum_{u=1}^K p_u f_u$, almost everywhere. Therefore, equality holds if and only if $f_1=f_2=\cdots=f_K$,  outside a set of Lebesgue measure $0$. 
\end{proof}

Now, suppose there exists $1 \leq s \ne t \leq K$ such that $f_s \ne f_t$ on a set of positive Lebesgue measure. By \eqref{eq:HPab} and Lemma \ref{componentwise}, it follows that 
$$\frac{R_{K, N}-\mathbb{E}_{H_0} (R_{K, N})}{N} \xrightarrow{a.s.} H(f_1,f_2,\ldots,f_k) - H(f,f,\ldots,f) < 0,$$ 
and hence, $\frac{1}{\sqrt N}(R_{K, N}-\mathbb{E}_{H_0}(R_{K, N})) \xrightarrow{a.s.} -\infty$. Now, since the matrix $\mathrm{Cov}_{H_0}(\vec{\bm A}_N)$ scales with $N$ (by Proposition \ref{ppn:mean_var}), $\lim_{N\rightarrow \infty} \frac{1}{N} \mathrm{Var}_{H_0}(R_{K, N}) < \infty$. Hence, 
\begin{equation}\label{cons1last}
\lim_{N\rightarrow\infty} Q_{K, N}=\lim_{N\rightarrow\infty} \frac{R_{K, N} - \mathbb{E}_{H_0}(R_{K, N})}{\sqrt{\mathrm{Var}_{H_0}(R_{K, N})}} \xrightarrow{a.s.} -\infty, 
\end{equation}
which implies the limiting power of the MCM test $\lim_{N \rightarrow \infty}\P_{H_1}(Q_{K, N} <  z_{\alpha})=1$, proving universal consistency.

For the MMCM test note that under the alternative,  $\frac{1}{N}(\vec{\bm A}_N - \mathbb{E}_{H_0}\vec{\bm A}_N) \xrightarrow{a.s.} \Delta _0 \in \R^{{K \choose 2}}$, 
where sum of the entries of $\gamma_0$ is $H(f_1,f_2,\ldots,f_K) - H(f,f,\ldots,f)$, that is,  $\Delta_0$ is non-zero. Now, since $\frac{1}{N} \mathrm{Cov}_{H_0}(\vec{\bm A}_N) \rightarrow \bm \Gamma$, where $\bm \Gamma$ is positive definite (by Proposition \ref{ppn:mean_var} and Lemma \ref{invert1}), 
$$\frac{1}{N}S_{K, N} \xrightarrow{a.s.} \Delta _0^\top \bm \Gamma^{-1} \Delta _0 > 0,$$  that is, $S_{K, N} \xrightarrow{a.s.} \infty$ under the alternative, proving universal consistency of the MMCM test.

\section{Proof of Theorem \ref{altclt}}\label{secalt}

Recall the alternative way to describe the joint distribution of the data described in Section \ref{condasm1}: Choose $\cZ_N:=(Z_1, Z_2, \ldots, Z_N)$ i.i.d. from the density $\phi_N=\sum_{s=1}^K \frac{N_s}{N} f_s$ in $\R^d$. Then given $\cZ_N=(Z_1, Z_2, \ldots, Z_N)$, assign a random label $L_j \in [K]:=\{1, 2, \ldots, K\}$ to $Z_j$, independently for each $1 \leq j \leq N$, as in \eqref{eq:labels}. Then it is easy to verify that the following fact, which is proved in Appendix \ref{psecalt1}.

\begin{lem}\label{condrel}
The joint distribution of $\left(\{Z_j: L_j=1\}, \{Z_j: L_j=2\}, \ldots, \{Z_j: L_j=K\}\right)$ conditional on $(\eta_1,\cdots,\eta_K)=(N_1, N_2, \ldots, N_K)$ is same as the joint distribution of the data $\left(\bm X^{(1)},\bm X^{(2)},\ldots,\bm X^{(K)}\right)$, where $\eta_s := \sum_{i=1}^N \bm 1 \{L_i = s\}$, the number of elements labelled $s$, for $1 \leq s \leq K$. 
\end{lem}

 The proof of Theorem \ref{altclt} has two steps: (1) Computing the conditional covariance matrix $\bm R(\cZ_N):=\mathrm{Cov}_{H_1}((\vec{\bm B}_N^\top, \eta_1,\ldots, \eta_{K-1})|\cZ_N)$~(as usual, $\vec{\bm B}_N$ denotes the vectorized upper triangular part of the matrix $\bm B_N$ defined in \eqref{eq:newbdefn}), under the bootstrap alternative distribution (which is the unconditional distribution of $\left(\{Z_j: L_j=1\}, \{Z_j: L_j=2\}, \ldots, \{Z_j: L_j=K\}\right)$), and show that scales with $N$ (Section \ref{sec:pfjointcovariance}), and (2) deriving the asymptotic normality of $\vec{\bm A}_N$ from the joint distribution of the vector $(\vec{\bm B}_N^\top, \eta_1,\ldots, \eta_{K-1})^\top$ under the bootstrap alternative distribution (Section \ref{sec:limitdistributionjoint}).

\subsection{Computing the Joint Conditional Covariance Matrix}
\label{sec:pfjointcovariance}

Given the $K$ densities $f_1, f_2, \ldots, f_K$, we will begin by defining the matrix $\Gamma_{f_1, f_2, \ldots, f_K}$ in Theorem \ref{altclt}. To this end, let $\phi=\sum_{s=1}^K p_s f_s$ and for each $s,t \in [K]$, define the function $h_{st}: \mathbb{R}^d \times \mathbb{R}^d \mapsto [0,1]$ as:
$$h_{st}(x,y) := \frac{p_sp_tf_s(x)f_t(y)}{ \phi(x) \phi(y) },$$
and set $\bar{h}_{st}(x,y) := h_{st}(x,y) + h_{st}(y,x)$. 

\begin{defn}\label{defn:GammaK}(Defining the matrix $\bm \Gamma_{f_1, f_2, \ldots, f_K}$) Throughout, let $Z \sim \phi=\sum_{s=1}^K p_s f_s$. To begin with, let $\bm Q$ be a square matrix of dimension $\binom{K}{2}+K-1$, partitioned as:
\begin{align}\label{eq:Q}
\bm Q=
\left[
\begin{array}{cc}
\bm Q_{11} & \bm Q_{12} \\
\bm Q_{12}^\top & \bm Q_{22}
\end{array}
\right],
\end{align}
where $\bm Q_{11}$, $\bm Q_{12}$, and $\bm Q_{22}$ have dimensions $\binom{K}{2}\times \binom{K}{2}$, $\binom{K}{2}\times (K-1)$ and $(K-1)\times(K-1)$, respectively, and their elements are defined as follows:  
\begin{itemize}

\item 
The elements of the matrix $\bm Q_{11}$ will be denoted by $q_{11}((s, t), (u, v))$, for $1\leq s < t\leq K$ and $1\leq u < v\leq K$, which is defined as: 
\begin{align}\label{eq:Qmain1}
q_{11}((s,t), (u,v)) :=   \left\{
\begin{array}{ll}
\frac{1}{2}\mathbb{E}\left[  \bar{h}_{st}(Z, Z) (1-\bar{h}_{st}(Z, Z) ) \right] & \textrm{if}~(s,t)=(u,v), \\ \\ 
-\frac{1}{2}\mathbb{E}\left[ \bar{h}_{st}(Z, Z)\bar{h}_{uv}(Z, Z) \right] & \textrm{otherwise}.
\end{array} 
\right.
\end{align}

\item 
The elements of the matrix $\bm Q_{12}$ will be denoted by $q_{12}((s, t), u)$, for $1\leq s < t\leq K$ and $1\leq u \leq K-1$, which is defined as:  
\begin{align}\label{eq:Qcross}
q_{12}((s,t),u) =   \left\{
\begin{array}{ll}
 \frac{1}{2} \E\left[\bar{h}_{st}(Z,Z)\left(1-\frac{2p_uf_u(Z)}{\phi(Z)}\right) \right] & \textrm{if}~u \in \{s,t\} , \\\\
- \E\left[ \bar{h}_{st}(Z, Z)\left(\frac{p_uf_u(Z)}{\phi(Z)}\right) \right] & \textrm{otherwise}. 
\end{array} 
\right.
\end{align}

\item The elements of the matrix $\bm Q_{22}$ will be denoted by $q_{22}(s, t)$, for $1\leq s, t\leq K-1$, which is defined as:  
\begin{align}\label{eq:Qmain2}
q_{22}(s, t) :=   \left\{
 \begin{array}{ll}
p_s(1-p_s)& \textrm{if}~s=t, \\ 
- p_s p_t& \textrm{otherwise}. 
 \end{array} 
 \right.
\end{align} 
\end{itemize}
Finally, define 
\begin{align}\label{eq:matrix_gamma}
\bm \Gamma_{f_1, f_2, \ldots, f_K} := \bm Q_{11} - \bm Q_{12} \bm Q_{22}^{-1} \bm Q_{12}^\top.
\end{align} 
\end{defn}

To compute the limit of the covariance matrix $\bm R(\cZ_N):=\mathrm{Cov}_{H_1}((\vec{\bm B}_N^\top, \eta_1,\ldots, \eta_{K-1})|\cZ_N)$, we need the following lemma from \cite{castropell}. Recall  $\phi:=\sum_{i=1}^s p_s f_s$.

\begin{lem}\label{plimh}\cite[Proposition 1]{castropell} Let $Z_1, Z_2, \ldots, Z_N$ be i.i.d. from the density $\phi_N=\sum_{s=1}^K \frac{N_s}{N} f_s$, and $g: \mathbb{R}^d \times \mathbb{R}^d \mapsto [0,1]$ be a symmetric, measurable function, such that almost any $z \in \mathbb{R}^d$ is a Lebesgue continuity point of $\phi(\cdot)g(z,\cdot)$. Then, as $N \rightarrow \infty$,
$$\frac{1}{N} \sum_{1\leq i<j\leq N} e(Z_i,  Z_j)g(Z_i, Z_j)~ \xrightarrow{P}~ \tfrac{1}{2} ~\mathbb{E} g(Z,Z),$$
where $e(x,y) := \bm 1 \{(x,y) \in  E(\sg(\cZ_N \cup \{x, y\})\}$ and $Z \sim \phi$. 
\end{lem}

The following lemma uses the above result to show that the conditional covariance matrix of $(\vec{\bm B}_N^\top, \eta_1,\ldots, \eta_{K-1})$ divided by $N$, converges to a deterministic limit in probability.


\begin{lem}\label{lm:covarianceH1}  Let $\bm R(\cZ_N):=\mathrm{Cov}_{H_1}((\vec{\bm B}_N^\top, \eta_1,\ldots, \eta_{K-1})|\cZ_N)$ be the conditional covariance matrix under the bootstrap alternative distribution. Then 
\begin{align}\label{eq:matrix_R}
\frac{1}{N} \bm R(\cZ_N) \pto \bm R:=
\left[
\begin{array}{cc}
\bm Q_{11} & \bm Q_{12} \\
\bm Q_{12}^\top & \bm R_{22}
\end{array}
\right],
\end{align}
where the matrices $\bm Q_{11}$ and $\bm Q_{12}$ are as defined in \eqref{eq:Qmain1} and \eqref{eq:Qcross}, respectively, and $\bm R_{22}=((r_{22}(s, t)))_{1\leq s,  t\leq K-1}$, where   
\begin{align}\label{eq:Rmain}
r_{22}(s, t) :=   \left\{
 \begin{array}{ll}
\E \left[ \frac{p_s f_s(Z)}{\phi(Z)}\left(1-\frac{ p_s f_s(Z)}{\phi(Z)}\right) \right] & \textrm{if}~s=t, \\ \\
- \E \left[\frac{ p_s  p_t f_s(Z) f_t(Z)}{\phi(Z)^2} \right] & \textrm{if}~s\neq t,  
\end{array} 
 \right.
\end{align}
with $Z \sim \phi:=\sum_{s=1}^K p_s f_s$, as before.  
\end{lem}\vspace{-0.1in}

\begin{proof} Let $\phi_N=\sum_{s=1}^K \frac{N_s}{N} f_s$, and recall that $\eta_s=\sum_{j=1}^N \bm 1 \{L_j=s\}$, for $s \in [K]$. This implies,  for $1 \leq s, \ne t \in [K-1]$, 
$$\mathrm{Var}_{H_1}(\eta_s|\cZ_N) = \sum_{j=1}^N  \frac{\frac{N_s}{N} f_s(Z_j)}{\phi_N(Z_j)}\left(1-\frac{ \frac{N_s}{N} f_s(Z_j)}{\phi_N(Z_j)}\right),$$ 
and $\mathrm{Cov}_{H_1}(\eta_s, \eta_t|\cZ_N) = - \sum_{j=1}^N \frac{ \frac{N_s}{N} \frac{N_t}{N} f_s(Z_j) f_t(Z_j)}{\phi_N(Z_j)^2}$. Hence, by the law of large numbers and the dominated convergence theorem, as $N \rightarrow \infty$, 
\begin{equation}\label{csc3}
\frac{1}{N}\Cov_{H_1}((\eta_1, \eta_2, \ldots, \eta_K)^\top|\cZ_N)  \xrightarrow{P} \bm R_{22},
\end{equation}
where $\bm R_{22}$ is as defined in \eqref{eq:Rmain}. 

Next, define $h_{st}^{(N)}(x,y) := \frac{\frac{N_s}{N}\frac{N_t}{N}f_s(x)f_y(t)}{\phi_N(x)\phi_N(y)}$, and let $\bar h_{st}^{(N)}(x,y) := h_{st}^{(N)}(x,y)  + h_{st}^{(N)}(x,y) $. Now, for each $1\leq s < t\leq K$, the conditional variance of $b_{st}$ (recall \eqref{eq:newbdefn}) is, 
\begin{align}
\frac{1}{N}\mathrm{Var}_{H_1}(b_{st}\big|\cZ_N) & = \frac{1}{N} \sum_{1\leq i < j \leq N} e(Z_s, Z_t) \bar h_{st}^{(N)}(Z_s, Z_t)(1-\bar h_{st}^{(N)}(Z_s, Z_t) )  \nonumber \\
\label{mz1} & \xrightarrow{P} \frac{1}{2}\mathbb{E}\left[  \bar{h}_{st}(Z, Z) (1-\bar{h}_{st}(Z, Z) ) \right], 
\end{align}
by Proposition \ref{plimh}, as  $\bar h_{st}^{(N)} \rightarrow \bar h_{st}$ uniformly. Similarly, for any two distinct pairs $(s, t)$ and $(u, v)$ with $1\leq s< t\leq K$ and $1\leq u< v\leq K$, \begin{align}\label{mz2}
\frac{1}{N}\mathrm{Cov}_{H_1}(b_{st}, b_{uv} \big|\cZ_N) 
& \xrightarrow{P} -\frac{1}{2}\mathbb{E}\left[ \bar{h}_{st}(Z, Z)\bar{h}_{uv}(Z, Z) \right].
\end{align}
Combining \eqref{mz1} and \eqref{mz2} gives 
\begin{equation}\label{csc1}
\frac{1}{N}\Cov_{H_1}({\bm {\vec B}}_N^\top|\cZ_N)  \xrightarrow{P} \bm Q_{11},
\end{equation}
where $\bm Q_{11}$ is as defined in \eqref{eq:Qmain1}.

Finally, for each $1\leq s, t \leq K$ and $1 \leq u \in K-1$, 
$$\mathrm{Cov}_{H_1}(b_{st}, \eta_u\big|\cZ_N) = \sum_{1\leq s < t\leq N}  e(Z_i,Z_j) \psi_{\{(s, t), u\}}^{(N)}(Z_i,Z_j),$$
where 
\[\psi_{\{(s, t), u\}}^{(N)}(Z_i,Z_j) =   \left\{
\begin{array}{ll}
\bar h^{(N)}_{st}(Z_i,Z_j)\left[1-\frac{\frac{N_u}{N} f_u(Z_i)}{\phi_N(Z_i)} - \frac{\frac{N_u}{N} f_u(Z_j)}{\phi_N(Z_j)}\right]& \textrm{if}~ u \in \{s, t\}), \\\\
-\bar h^{(N)}_{st}(Z_i,Z_j) \left[\frac{\frac{N_u}{N} f_u(Z_i)}{\phi_N(Z_i)} + \frac{\frac{N_u}{N} f_u(Z_j)}{\phi_N(Z_j)}\right] & \textrm{if}~u \notin \{s, t\}. \\
\end{array} 
\right. \]
Then by Lemma \ref{plimh}, as $N\rightarrow \infty$,
\begin{equation}\label{csc2}
\frac{1}{N} \mathrm{Cov}_{H_1}(b_{st}, \eta_u\big|\cZ_N ) \xrightarrow{P} q_{12}((s, t), u), 
\end{equation}
where $q_{12}((s, t), u)$ is as in \eqref{eq:Qcross}. 

The result in \eqref{eq:matrix_R} now follows by combining \eqref{csc3}, \eqref{csc1}, and \eqref{csc2}, completing the proof of the lemma. 
\end{proof}\vspace{-0.1in}

\subsection{The Joint Central Limit Theorem of the Cross-Counts}
\label{sec:limitdistributionjoint}
 We now have all the tools necessary for proving Theorem \ref{altclt}. Towards  this, define
$$V_N := \frac{1}{\sqrt N} \left(\vec{\bm B}_N^\top-\mathbb{E}_{H_1}(\vec{\bm B}_N^\top\big|\cZ_N),~\eta_1 - \mathbb{E}_{H_1}(\eta_1\big| \cZ_N),~\cdots,~\eta_{K-1} - \mathbb{E}_{H_1}(\eta_{K-1}\big| \cZ_N)\right)^\top,$$ 
a vector of length ${K \choose 2} + K-1$. Define $U_N := \overline{\bm R}(\cZ_N)^{-\frac{1}{2}}V_N$, where  $\overline{\bm R}(\cZ_N)= \frac{1}{N}\bm R(\cZ_N)$, so that $\mathbb{E}_{H_1}(U_N|\cZ_N) = 0$ and $\mathrm{Cov}_{H_1}(U_N \big|\cZ_N) = \mathrm I$, under the bootstrap alternative distribution.

For each $(a,b) \in [N]^2$, define the $K \times K$ matrix $\bm C_{ab}=((C_{ab}(s, t)))_{1\leq s \ne t\leq K}$, where, 
$$C_{ab}(s, t) := \bm 1 \left\{\{L_a, L_b\}=\{s,t\}\right\},$$
for $s \ne t$ and zero otherwise. Let $\bm {\vec C}_{ab}$ be the vector of length ${K \choose 2}$ obtained by concatenating the rows of $\bm C_{ab}$ in the upper triangular part. 
Now, for each $(a,b)\in [N]^2$, define:
$$Y_{ab} := \big(\bm {\vec C}_{ab}^\top,~\bm 1 \{L_a=1\} + \bm 1 \{L_b = 1\},~\cdots,~\bm 1 \{L_a=K-1\} + \bm 1 \{L_b = K-1\}\big)^\top,$$
and let $\overline{Y}_{ab} := \frac{1}{\sqrt N} \overline{\bm R}(\cZ_N)^{-\frac{1}{2}} ({Y}_{ab}- \mathbb{E}_{H_1}({Y}_{ab}|\cZ_N ))$.  Further, define $S(\cZ_N) := \{\{a,b\} \subset [N] : e(Z_a,Z_b) = 1\}$. Then, it is easy to see that:
$$U_N = \sum_{\{a,b\} \in S(\cZ_N)} \overline{Y}_{ab}.$$ 
Note that, under the bootstrap alternative distribution, given $\cZ_N$, the collection $\{\overline{Y}_{ab}\}_{\{a,b\} \in S(\cZ_N)}$ is independent, so by an application of the multivariate Berry-Essen theorem \cite[Theorem 1.1]{raic}, we get:
\begin{equation}\label{raic1}
\sup_{A \in \mathcal{C}} \Big|\mathbb{P}(U_N \in A\big| \cZ_N) -  \Phi_{{K \choose 2}+K-1}(A) \Big| \leq L(K) \sum_{\{a, b\} \in S(\cZ_N)} \mathbb{E}\left(||\overline{Y}_{ab}||^3\big| \cZ_N\right),
\end{equation}
where $\mathcal{C}$ denotes the class of all measurable convex subsets of $\mathbb{R}^{{K \choose 2}+K-1}$, $\Phi_{{K \choose 2}+K-1}(\cdot)$ the standard normal distribution function in dimension ${K \choose 2}+K-1$, and $L(K)$ is a constant depending only on $K$.

\begin{lem}\label{spnorm} Let $\overline{Y}_{ab}$ be as defined above. Then $\sum_{\{a, b\} \in S(\cZ_N)} \mathbb{E}\left(||\overline{Y}_{ab}||^3\big| \cZ_N\right) \xrightarrow{P} 0$, as $N \rightarrow \infty$.
\end{lem}

\begin{proof} 
Note that every entry of the vector $Y_{ab} - \mathbb{E}_{H_1}(Y_{ab}|\cZ_N)$ is bounded in absolute value by $2$. Hence, $\|Y_{ab} - \mathbb{E}_{H_1}(Y_{ab}| \cZ_N)\| \leq \sqrt{2(K-1)(K+2)}$. Consequently, denote the operator norm of matrix by $||\cdot||_{\mathrm{op}}$
$$\|\overline{Y}_{ab}\| \leq \frac{1}{\sqrt N}\left\|\overline{\bm R}(\cZ_N)^{-\frac{1}{2}}\right\|_{\mathrm{op}}\left\|Y_{ab} - \mathbb{E}\left(Y_{ab}\big|\cZ_N\right)\right\|\leq \left(\frac{2(K-1)(K+2)}{N}\right)^\frac{1}{2}\left\|\overline{\bm R}({\cZ_N})^{-\frac{1}{2}}\right\|_{\mathrm{op}}. $$
Then by the Cauchy-Schwarz inequality, 
\begin{eqnarray*}
\sum_{\{a, b\} \in S(\bm Z)} \mathbb{E}\left(||\overline{Y}_{ab}||^3\big| \cZ_N \right) &\leq& \sqrt{\frac{2(K-1)^3(K+2)^3}{N}}\left\|\overline{\bm R}({\cZ_N})^{-\frac{1}{2}}\right\|_{\mathrm{op}}^3\\ &\leq& \sqrt{\frac{2(K-1)^3(K+2)^3}{N}} \left[\mathrm{tr}\left(\overline{\bm R}({\cZ_N})^{-\frac{1}{2}}\right)\right]^3. 
\end{eqnarray*}
The RHS above converges to zero in probability, because by Lemma \ref{lm:covarianceH1} $\overline{\bm R}({\cZ_N})$ converges in probability.
\end{proof}

The lemma combined with \eqref{raic1} shows that, under the bootstrap alternative distribution, the vector $U_N|\cZ_N$ converges in distribution to $N_{{K \choose 2}+K-1}(0, \mathrm I)$, and by Lemma \ref{lm:covarianceH1} $V_N|\cZ_N$ converges in distribution to $N_{{K \choose 2}+K-1}(0, \bm R)$, where $\bm R$ is as defined in \eqref{eq:matrix_R}. 
Hence, for every vector $t \in \mathbb{R}^{{K \choose 2}+K-1}$, 
\begin{align}\label{eq:mgfVN}
\mathbb{E}\left(e^{it^\top V_N}\Big| \cZ_N\right) \xrightarrow{P} \mathbb{E}\left(e^{it^\top W}\right),
\end{align} 
where $W \sim N_{{K \choose 2}+K-1}(0, \bm R)$. Next, define 
$$B_N := \frac{1}{\sqrt N}\left(\bm 0^\top,~\mathbb{E}_{H_1}(\eta_1\big|\cZ_N)-N_1,~\cdots,~\mathbb{E}_{H_1}(\eta_{K-1}\big|\cZ_N) - N_{K-1}\right)^\top,$$
where the $\bm 0$ here denotes a vector of all zeros of length ${K \choose 2}$.  Now, by the usual central limit theorem, under the bootstrap alternative distribution, as $N \rightarrow \infty$ , $B_N \xrightarrow{D} N_{{K \choose 2}}(0, \bm \Psi)$, where 
$$\bm \Psi:=
\left[
\begin{array}{cc}
\bm 0 & \bm 0\\
\bm 0 & \bm M 
\end{array}
\right],
$$
with the elements of $\bm M$ will be denoted by $m(s, t)$, for $1\leq s, t\leq K-1$, and 
\begin{align}\label{eq:Rmain1}
m(s, t) :=   \left\{
 \begin{array}{ll}
\Var_{Z \sim \phi}\left[ \frac{p_s f_s(Z)}{\phi(Z)} \right] & \textrm{if}~s=t, \\ \\
\Cov_{Z \sim \phi}\left[ \frac{p_s f_s(Z)}{\phi(Z)}, \frac{p_t f_t(Z)}{\phi(Z)} \right]  & \textrm{if}~s\neq t.  
\end{array} 
 \right.
\end{align}

Now, recalling the definitions of the matrix $\bm Q$ (from \eqref{eq:Q}) and the matrix $\bm R$ (from \eqref{eq:matrix_R}), it is easy to see that $\bm Q=\bm R+ \bm \Psi$.
Hence, by \eqref{eq:mgfVN} and Lemma \ref{difflem} (putting $A_N = V_N$, $\bm C_N = \cZ_N$ and $f_N(\bm C_N) = B_N$), gives
\begin{align*}
\frac{1}{\sqrt N}\left(\vec{\bm B}_N^\top-\mathbb{E}_{H_1}(\vec{\bm B}_N^\top|\cZ_N),~ \Delta_N^\top \right)^\top & =V_N + B_N \nonumber \\
&\xrightarrow{D} N_{{K \choose 2}}(\bm 0, \bm R+ \bm \Psi)\stackrel{D}=N(\bm 0, \bm Q). 
\end{align*}
where $\Delta_N:=(\eta_1-N_1,~\cdots,~\eta_{K-1}-N_{K-1})^\top$. Therefore, by Lemma \ref{condrel} and \eqref{eq:matrix_gamma}, the distribution of $\frac{1}{\sqrt N}(\vec{\bm B}_N^\top-\mathbb{E}_{H_1}(\vec{\bm B}_N^\top|\cZ_N))$ conditional on $\Delta_N =\bm 0$, converges to $N_{{K \choose 2}}(\bm 0,  \bm Q_{11} - \bm Q_{12} \bm Q_{22}^{-1} \bm Q_{12}^\top)$, as required.

\section{Proofs of Technical Lemmas}
\label{sec:technical_lemmas}

Here, we collect the proofs of the different technical lemmas, required in the proofs above. In Section \ref{sec:invertcov}, we show the invertibility of the matrix $\mathrm{Cov}_{H_0}(\vec{\bm A}_N)$. The proof of Lemma \ref{condrel} is given in Section \ref{psecalt1}. Other technical lemma used in the proof of Theorem \ref{secalt} are proved in Section \ref{sec:technical_lemma_CLT}. 

\subsection{Invertibility of the Count Matrix Under the Null}
\label{sec:invertcov}
 
In order for the MCMM statistic to be well-defined, we need to make sure the matrix $\mathrm{Cov}_{H_0}(\vec{\bm A}_N)$ (recall Proposition \ref{ppn:mean_var}) is invertible. This is proved in the following lemma: 

\begin{lem}\label{invert1} The matrix $\mathrm{Cov}_{H_0}(\vec{\bm A}_N)$ is invertible.
\end{lem}

\begin{proof} For simplicity, we assume that the sample sizes $N_s$ are even, for all $1 \leq s \leq K$. For $1 \leq r < s \leq K$, define a $K\times K$ matrices $\bm a_{rs}= ((a_{rs}(u, v)))_{1\leq u, v \leq K}$ as follows:
$$ 
\begin{array}{lc}
a_{rs}(r, s)=a_{rs}(s, r) = \min \{N_r,N_s\},  &  \\
a_{rs}(u, v)= 0, & \text{for all } \{u, v \}\neq \{r,s\} \text{ and } u \neq v, \\ 
a_{rs}(r, r) = \frac{1}{2}(N_r - \min\{N_r,N_s\}),  &     \\
a_{rs}(s, s) = \frac{1}{2}(N_s - \min\{N_r,N_s\} ),  &   \\
a_{rs}(u, u)  = \frac{1}{2} N_u, & \text{ for all } u \neq \{r,s\}.  
\end{array}
$$ 
Clearly, $\bm a_{rs} \in \sB$ (recall Proposition \ref{ppn:H0_distribution}), which implies, by \eqref{nulld}, $\mathbb{P}_{H_0}(\bm A_N=\bm a_{rs}) > 0$, for all $1\leq r < s \leq K$. Now, as in \eqref{eq:SKN}, denote by $\vec{\bm a}_{rs}$ the vector of length ${K \choose 2}$ obtained by concatenating the rows of $\bm a_{rs}$ in the upper triangular part. The argument above shows that $\mathbb{P}_{H_0}(\vec{\bm A}_N= \vec{\bm a}_{rs}) > 0$, for all $1\leq r<s\leq K$. Moreover, also note that $\mathbb{P}_{H_0}(\vec{\bm A}_N=\bm 0) > 0$, where $\bm 0$ denotes the vector of length ${\binom{K}{2}}$ with all entries $0$. Also, note that the vectors $\{\vec{\bm a}_{rs} \}_{1\leq r<s\leq K}$, each of which has only have one non-zero element corresponding to the element $a_{rs}(r,s)$, form a basis of $\mathbb{R}^{\binom{K}{2}}$.

Now, suppose that $\mathrm{Cov}_{H_0}(\vec{\bm A}_N)$ is singular, whence there exists a non-zero vector $\eta \in \mathbb{R}^{\binom{K}{2}}$ such that $\mathrm{Cov}_{H_0}(\vec{\bm A}_N) \eta = 0$. This implies that $\mathrm{Var}_{H_0}(\eta^\top \vec{\bm A}_N) = 0$, and hence, 
\begin{align}\label{eq:thetaAN}
\mathbb{P}_{H_0}\left(\eta^\top \vec{\bm A}_N = \mathbb{E}_{H_0} (\eta^\top \vec{\bm A}_N)\right) = 1.
\end{align}
The fact $\mathbb{P}_{H_0}(\vec{\bm A}_N=0) > 0$, now implies that $\mathbb{E}_{H_0} (\eta^\top \vec{\bm A}_N) = 0$ (otherwise, assuming  $\mathbb{E}_{H_0} (\eta^\top \vec{\bm A}_N) \ne 0$, leads to, by \eqref{eq:thetaAN}, $\mathbb{P}_{H_0}(\vec{\bm A}_N=0) \leq \mathbb{P}_{H_0}(\eta^\top \vec{\bm A}_N = 0) = 0$, which is a contradiction). Again, since $\mathbb{P}_{H_0}(\vec{\bm A}_N= \vec{\bm a}_{rs}) > 0$, it follows that $\eta^\top \vec{\bm a}_{rs} = 0$, for all $1\leq r<s\leq K$.  This implies, since the vectors $\{\vec{\bm a}_{rs} \}_{1\leq r<s\leq K}$ form a basis of $\mathbb{R}^{\binom{K}{2}}$, $\eta = \bm 0$, which is a contradiction.
\end{proof}

\subsection{Proof of Lemma \ref{condrel}}\label{psecalt1}

For notational convenience, we will prove the result only for the case $K=2$. The proof for general $K$ follows similarly.  We begin with a few notations: Let $\Pi$ denote the set of all permutations $\sigma$ of $[N]$, such that $\sigma(a) < \sigma(b)$, for all $1 \leq a < b \leq N_1$, and $\sigma(a) < \sigma(b)$, for all $N_1 + 1 \leq a < b \leq N_1+N_2=N$. Moreover, for a vector $x \in \R^d$ and a $\R^d$-valued random variable $X$, we denote by $\{X \leq x\}$ the event $\{X \in \{y \in \mathbb{R}^d: y \leq x\}\}$.\footnote{For any two vectors $u=(u_1, u_2, \ldots, u_d) \in \R^d$ and $v=(v_1, v_2, \ldots, v_d) \in \R^d$, we write $u \leq v$, if $u_a \leq  v_a$, for all $1 \leq a \leq d$.}

Now, considering the sets $\{Z_i: L_i=1\}$ and  $\{Z_i: L_i=2\}$ as vectors with the indices arranged in increasing order, it follows that 
$$(\{Z_i: L_i=1\}, \{Z_i: L_i=2\})=(Z_{\pi(1)}, \cdots, Z_{\pi(N)}),$$ 
where $\pi$ is a random permutation of $[N]$, such that $\pi(1)< \pi(2) \cdots<\pi(\eta_1)$ and $\pi(\eta_1+1)< \pi(\eta_1+2) \cdots<\pi(\eta_1+\eta_2)$, where  $L_{\pi(a)} = 1$, for all $1 \leq a \leq \eta_1$, and $L_{\pi(a)} = 2$, for all $\eta_1+1 \leq a \leq \eta_1+ \eta_2$. Then, for every $z_1,\cdots,z_N \in \mathbb{R}^d$, 
\begin{align*}
&\mathbb{P}\left(Z_{\pi(1)}  \leq z_1, \cdots, Z_{\pi(N)}  \leq z_N  \Bigg| \eta_1 = N_1\right)\\
&= \frac{1}{\mathbb{P}(\eta_1 = N_1)}\sum_{\sigma \in \Pi} \mathbb{P}\left(Z_{\sigma(1)} \leq z_1, \cdots, Z_{\sigma(N)} \leq z_N,~\pi = \sigma,~\eta_1 = N_1\right)\\
&= \frac{1}{\mathbb{P}(\eta_1 = N_1)}\sum_{\sigma \in \Pi}\mathbb{P}\left(\bigcap_{a=1}^{N_1} \left\{Z_{\sigma(a)} \leq z_a ,~L_{\sigma(a)} = 1\right\} \bigcap \bigcap_{a=N_1+1}^{N_1+N_2} \left\{Z_{\sigma(a)}  \leq z_a ,~L_{\sigma(a)} = 2 \right\}\right)\\
&=\frac{|\Pi|}{\mathbb{P}(\eta_1 = N_1)}  \prod_{a=1}^{N_1}\mathbb{P}\left(Z_{a} \leq z_a ,~L_{a} = 1\right) \prod_{a=N_1+1}^{N_1+N_2}\mathbb{P}\left(Z_{a} \leq z_a,~L_{a} = 2\right)\\
&= \frac{\binom{N}{N_1}}{\mathbb{P}(\eta_1 = N_1)}  \prod_{a=1}^{N_1} \frac{N_1}{N} F_1(z_a) \prod_{a=N_1+1}^{N_1+N_2} \frac{N_2}{N} F_2(z_a)\\
&=  \prod_{a=1}^{N_1} F_1(z_a) \prod_{a=N_1+1}^{N_1+N_2} F_2(z_a) \tag*{(using $\mathbb{P}(\eta_1 = N_1)=\binom{N}{N_1} \left(\frac{N_1}{N}\right)^{N_1} \left(\frac{N_2}{N}\right)^{N_2}$)} \\
&=\P(X_1^{(1)} \leq z_1, \cdots,  X_{N_1}^{(1)} \leq z_{N_1}, X_1^{(2)} \leq z_{N_1+1}, \cdots,  X_{N_2}^{(2)} \leq z_{N_1+N_2}),
\end{align*}
which completes the proof of the lemma. \qed

\subsection{Missing Details in the Proof of Theorem \ref{altclt}}
\label{sec:technical_lemma_CLT}

Here, we provide the proof of a lemma used in the proof of Theorem \ref{altclt}.


\begin{lem}\label{difflem} 
Let $\{X_N\}_{N\geq 1}$ be a sequence of $\mathbb{R}^p$-valued random vectors, for some $p\geq 1$, and $\bm C_N$ be a sequence random variable, such that $\mathbb{E}(e^{it^\top X_N}|\bm C_N) \pto a$, for some real number $a$ and some vector $t \in \mathbb{R}^p$. Moreover, suppose that $f_N$ is a sequence of deterministic functions with codomain $\mathbb{R}^p$, such that $\mathbb{E}(e^{it^\top f_N(\bm C_N)})\rightarrow b$, for some real number $b$. Then, $$\lim_{N \rightarrow \infty}\mathbb{E}\left(e^{it^\top (A_N + f_N(\bm C_N))}\right) = ab. $$
\end{lem}

\begin{proof} Note that,
\begin{align*}
\Big|\mathbb{E}\left(e^{it^\top (A_N + f_N(\bm C_N))}\right) - ab\Big| &= \Big|\mathbb{E}\left[e^{it^\top f_N(\bm C_N)}\mathbb{E}\left(e^{it^\top A_N}\big|\bm C_N\right)\right] - ab\Big|\\
&\leq \Big|\mathbb{E}\left[e^{it^\top f_N(\bm C_N)}\left(\mathbb{E}\left(e^{it^\top A_N}\big|\bm C_N\right)-a\right)\right]\Big| + |a|\Big|\mathbb{E}\left[e^{it^\top f_N(\bm C_N)} - b\right]\Big|\\
&\leq \mathbb{E}\Big|\mathbb{E}\left(e^{it^\top A_N}\big|\bm C_N\right)-a\Big| + |a|\Big|\mathbb{E}\left(e^{it^\top f_N(\bm C_N)} - b\right)\Big|. 
\end{align*}
The first term in the last expression goes to $0$ by hypothesis and the dominated convergence theorem, while the last term goes to $0$ by hypothesis, completing the proof.
\end{proof}

\section{Additional Simulations}\label{sec:lognormalapp}

\noindent In this section, we present simulations comparing the empirical powers of the MCM and the MMCM tests for location, spherical scale, and equi-correlation scale changes in the log normal family. As before, in all the simulations, the power is calculated over 100 iterations,  the tests are implemented using the permutation distribution,  and the nominal level is chosen to be $0.05$. 

\begin{table}[h]
	\centering
	\begin{minipage}[c]{0.59\textwidth}
		\centering
		\small{
			\begin{table}[H]
				\begin{tabular}{c|c||ccccccc}
					\hline
					$\Delta\downarrow$ & Dimension &  5 & 10 & 50 & 100& 200 & 300 & 500 \\  
					\hline
					\hline
					\multirow{2}{*}{.06}&MCM & \bf .13 & .17 & \bf .49 & .65 & .81 & .88 & .90 \\
					&MMCM & .10  & \bf .18 & .42 & \bf .68 & \bf .93 & \bf .98 & \bf  .99\\
					\hline
					\multirow{2}{*}{.08}&MCM & \bf .15 & .19 & .52 & .68 & .87 & .95 & 1.0 \\
					&MMCM & .12  & \bf .20 & \bf .70 & \bf .92 & \bf .98 & \bf 1.0 & \bf 1.0 \\
					\hline
					\multirow{2}{*}{.10}&MCM & \bf .19 & .26 & .79 & .85 & .96 & 1.0 & 1.0 \\
					&MMCM & .18  & \bf .31 & \bf .94 & \bf 1.0 & \bf 1.0 & \bf 1.0 & \bf 1.0\\
					\hline
					\multirow{2}{*}{.12}&MCM & \bf .50 & .69 & 1.0 & 1.0 & 1.0 & 1.0 & 1.0 \\
					&MMCM & .49  & \bf .84 & \bf 1.0 & \bf 1.0 &  \bf 1.0 & \bf 1.0 & \bf 1.0 \\
					\hline
				\end{tabular}
			\end{table}
			\vspace{-0.1in}
			(a)
		}
	\end{minipage}
	\begin{minipage}[c]{0.39\textwidth}
		\centering
		\small{
			\begin{table}[H]
				\begin{tabular}{c|c||cccc}
					\hline
					$\Delta\downarrow$ & Groups & 4  & 6 & 8 & 10  \\  
					\hline
					\hline
					\multirow{2}{*}{.04}&MCM & \bf .13 & \bf .22 & .49 & .71  \\
					&MMCM & .10  & .16 & \bf .63 & \bf .99 \\
					\hline
					\multirow{2}{*}{.06}&MCM & \bf .21 & .35 & .85 & 1.0  \\
					&MMCM & .14  & \bf .41 & \bf .98 & \bf 1.0  \\
					\hline
					\multirow{2}{*}{.08}&MCM & \bf .37 & .92 & 1.0 & 1.0  \\
					&MMCM & .31  & \bf 1.0 & \bf 1.0 & \bf 1.0 \\
					\hline
					\multirow{2}{*}{.10}&MCM & .52 & 1.0 & 1.0 & 1.0 \\
					&MMCM & \bf .56  & \bf 1.0 & \bf 1.0 & \bf 1.0  \\
					\hline
				\end{tabular}
			\end{table}
		}
		\vspace{-0.1in}
		(b)
	\end{minipage}
	\caption{\small{Power of the MCM and the MMCM tests in the lognormal location family with (a) the number of classes $K=6$ fixed, and (b) the dimension $d=150$ fixed.}}
	\label{table:lognormallocation}
	\vspace{-0.1in}
\end{table}

\begin{itemize}

\item {\it Lognormal Location}: Here, we consider samples from the following $K$ log-normal distributions: $\exp(N_d((s-1)\Delta \cdot \bm 1, \mathrm I))$, for $1 \leq s \leq K$. Table \ref{table:lognormallocation}(a) shows the fixed class scenario, where we take $K=6$ groups and vary the dimension $d$ from $5$ to $500$, and $\Delta$ from $0.06$ to $0.12$. Table \ref{table:lognormallocation}(b) shows the fixed dimension scenario, where the dimension $d = 150$ is fixed, the number of groups $K$ varies along $4,6,8,10$, and $\Delta$ varies from  $0.04$ to $0.10$. In both cases, the sample sizes  were taken in equal increments of $50$, starting from $50$.

\begin{table}[h]
	\centering
	\begin{minipage}[c]{0.59\textwidth}
		\centering
		\small{
			\begin{table}[H]
				\begin{tabular}{c|c||ccccccc}
					\hline
					$\Delta\downarrow$ & Dimension &  5 & 10 & 50 & 100& 200 & 300 & 500 \\  
					\hline
					\hline
					\multirow{2}{*}{.15}&MCM & \bf .15 & \bf .21 & .44 & .70 & .94 & .99 & 1.0 \\
					&MMCM & .13  & .17 & \bf .66 & \bf .99 & \bf 1.0 & \bf 1.0 & \bf 1.0\\
					\hline
					\multirow{2}{*}{.20}&MCM & .16 & .22 & .70 & .96 & 1.0 & 1.0 & 1.0 \\
					&MMCM & \bf .18  & \bf .35 & \bf .96 & \bf 1.0 & \bf 1.0 & \bf 1.0 & \bf  1.0\\
					\hline
					\multirow{2}{*}{.25}&MCM & \bf .17 & .23 & .87 & .97 & 1.0 & 1.0 & 1.0 \\
					&MMCM & .10  & \bf .37 & \bf .98 & \bf 1.0 & \bf 1.0 & \bf 1.0 & \bf 1.0 \\
					\hline
					\multirow{2}{*}{.30}&MCM & .18 & .27 & .91 & 1.0 & 1.0 & 1.0 & 1.0 \\
					&MMCM & \bf .21  & \bf .46 & \bf 1.0 & \bf 1.0 & \bf 1.0 & \bf 1.0 & \bf 1.0\\
					\hline
					\multirow{2}{*}{.35}&MCM & .13 & .42 & 1.0 & 1.0 & 1.0 & 1.0 & 1.0 \\
					&MMCM & \bf .20  & \bf .51 & \bf 1.0 & \bf 1.0 & \bf  1.0 & \bf 1.0 & \bf 1.0 \\
					\hline
					\multirow{2}{*}{.40}&MCM & \bf .34 & .73 & 1.0 & 1.0 & 1.0 & 1.0 & 1.0 \\
					&MMCM & .25  & \bf .96 & \bf 1.0 & \bf 1.0 & \bf  1.0 & \bf 1.0 & \bf 1.0 \\
					\hline
				\end{tabular}
			\end{table}
			\vspace{-0.1in}
			(a)
		}
	\end{minipage}
	\begin{minipage}[c]{0.39\textwidth}
		\centering
		\small{
			\begin{table}[H]
				\begin{tabular}{c|c||cccc}
					\hline
					$\Delta\downarrow$ & Groups & 4  & 6 & 8 & 10  \\  
					\hline
					\hline
					\multirow{2}{*}{.15}&MCM & .39 & .91 & .99 & 1.0  \\
					&MMCM & \bf .48  & \bf 1.0 & \bf 1.0 & \bf 1.0 \\
					\hline
					\multirow{2}{*}{.20}&MCM & .75 & .99 & 1.0 & 1.0  \\
					&MMCM & \bf .86  & \bf 1.0 & \bf 1.0 & \bf 1.0 \\
					\hline
					\multirow{2}{*}{.25}&MCM & .93 & 1.0 & 1.0 & 1.0  \\
					&MMCM & \bf .99  & \bf 1.0 & \bf 1.0 & \bf 1.0  \\
					\hline
					\multirow{2}{*}{.30}&MCM & .98 & 1.0 & 1.0 & 1.0  \\
					&MMCM & \bf 1.0  & \bf 1.0 & \bf 1.0 & \bf 1.0 \\
					\hline
					\multirow{2}{*}{.35}&MCM & .99 & 1.0 & 1.0 & 1.0 \\
					&MMCM & \bf 1.0  & \bf 1.0 & \bf 1.0 & \bf 1.0  \\
					\hline
					\multirow{2}{*}{.40}&MCM & 1.0 & 1.0 & 1.0 & 1.0 \\
					&MMCM & \bf 1.0  & \bf 1.0 & \bf 1.0 & \bf 1.0  \\
					\hline
				\end{tabular}
			\end{table}
		}
		\vspace{-0.1in}
		(b)
	\end{minipage}
	\caption{\small{Power of the MCM and the MMCM tests in the lognormal spherical scale family with (a) the number of classes $K=6$ fixed, and (b) the dimension $d=150$ fixed.}}
	\vspace{-0.1in}
	\label{table:lognormalscale}
\end{table}

\begin{table}[h]
	\centering
	\begin{minipage}[c]{0.59\textwidth}
		\centering
		\small{
			\begin{table}[H]
				\begin{tabular}{c|c||ccccccc}
					\hline
					$\Delta\downarrow$ & Dimension &  5 & 10 & 50 & 100& 200 & 300 & 500 \\  
					\hline
					\hline
					\multirow{2}{*}{.15}&MCM & \bf .08 & .09 & \bf .14 & \bf .16 & \bf .23 & \bf .31 & \bf .32 \\
					&MMCM & .05  & \bf .10 & .11 & .14 & .18 & .23 & .27\\
					\hline
					\multirow{2}{*}{.20}&MCM & \bf .09 & .11 & \bf .16 & .13 & .24 & .25 & .34 \\
					&MMCM & .08  & \bf .11 & .13 & \bf .22 & \bf .25 & \bf .27 &  \bf .39\\
					\hline
					\multirow{2}{*}{.25}&MCM & \bf .07 & \bf .11 & \bf .17 & \bf .26 & .30 & .36 & .43 \\
					&MMCM & .02  & .08 & .12 & .25 & \bf .33 & \bf .45 & \bf .53 \\
					\hline
					\multirow{2}{*}{.30}&MCM & .06 & .14 & .15 & .26 & .38 & .45 & .52 \\
					&MMCM & \bf .10  & \bf .15 & \bf .23 & \bf .36 & \bf .54 & \bf .68 & \bf .71\\
					\hline
					\multirow{2}{*}{.35}&MCM & .07 & .15 & .28 & .32 & .45 & .57 & .72 \\
					&MMCM & \bf .08  & \bf .16 & \bf .28 & \bf .43 & \bf  .69 & \bf .77 & \bf .89 \\
					\hline
					\multirow{2}{*}{.40}&MCM & \bf .20 & .22 & .37 & .45 & .72 & .75 & .84 \\
					&MMCM & .16  & \bf .24 & \bf .48 & \bf .63 & \bf  .95 & \bf .99 & \bf 1.0 \\
					\hline
				\end{tabular}
			\end{table}
			\vspace{-0.1in}
			(a)
		}
	\end{minipage}
	\begin{minipage}[c]{0.39\textwidth}
		\centering
		\small{
			\begin{table}[H]
				\begin{tabular}{c|c||cccc}
					\hline
					$\Delta\downarrow$ & Groups & 4  & 6 & 8 & 10  \\  
					\hline
					\hline
					\multirow{2}{*}{.15}&MCM & \bf .20 & .14 & .12 & .11 \\
					&MMCM & .14  & \bf .15 & \bf .12 & \bf .11 \\
					\hline
					\multirow{2}{*}{.20}&MCM & \bf .28 & .18 & .13 & .14  \\
					&MMCM & .22  & \bf .19 & \bf .18 & \bf .16 \\
					\hline
					\multirow{2}{*}{.25}&MCM & .31 & .31 & \bf .27 & \bf .21  \\
					&MMCM & \bf .37  &\bf  .32 & .19 & .16  \\
					\hline
					\multirow{2}{*}{.30}&MCM & .41 & .35 & .29 & .23  \\
					&MMCM & \bf .42  & \bf .40 & \bf .31 & \bf .30 \\
					\hline
					\multirow{2}{*}{.35}&MCM & .47 & .39 & .38 & .33 \\
					&MMCM & \bf .55  & \bf .51 & \bf .48 & \bf .45  \\
					\hline
					\multirow{2}{*}{.40}&MCM & .63 & .52 & .49 & .47 \\
					&MMCM & \bf .79 & \bf .77 & \bf .68 & \bf .63  \\
					\hline
				\end{tabular}
			\end{table}
		}
		\vspace{-0.1in}
		(b)
	\end{minipage}
	\caption{\small{Power of the MCM and the MMCM tests in the lognormal equi-correlated scale family with (a) the number of classes $K=6$ fixed, and (b) the dimension $d=150$ fixed.}}
	\label{table:lognormalequicorrelation}
	\vspace{-0.2in}
\end{table}

\item {\it Spherical Lognormal Scale}: Here, we consider samples from the following $K$ log-normal distributions: $\exp(N_d(\bm 0, (1+(s-1)\Delta) \mathrm I))$, for $1 \leq s \leq K$. Table \ref{table:lognormalscale}(a) shows the  fixed class scenario, with $K=6$ and dimension $d$ varying from $5$ to $500$, and $\Delta$ varying from $0.15$ to $0.4$.  As before, in this case, the sample sizes  were taken in equal increments of $50$, starting from $50$. 
Table \ref{table:lognormalscale}(b) shows the  fixed dimension scenario, where the $d = 150$ is fixed, and $K$ varies along $4,6,8,10$, and $\Delta$ varies from $0.15$ to $0.4$, as well. Here, the sample sizes are taken in equal increments from $50$ to $200$ when $K=4$, from $50$ to $300$ when $K=6$, from $50$ to $260$ when $K=8$, and from $50$ to $230$ when $K=10$.

\item {\it Equi-correlated Lognormal Scale}:  Here, we consider samples from the following $K$ log-normal distributions: $\exp(N_d(0, (1-\rho_s)\mathrm I + \rho_s \bm 1 \bm 1^\top))$, where $\rho_s := (s-1)\frac{\Delta}{K-1}$, for $1 \leq s \leq K$.   Table \ref{table:lognormalequicorrelation}(a) shows the  fixed class scenario, with $K=6$ and dimension $d$ varying from $5$ to $500$, and $\Delta$ varying from $0.15$ to $0.4$. The sample sizes are taken to be $50, 100, 150, 200, 250$ and $300$.
Table \ref{table:lognormalequicorrelation}(b) shows the  fixed dimension scenario, where the $d = 150$ is fixed, and $K$ varies along $4,6,8,10$, and $\Delta$ varies from $0.15$ to $0.4$, as before. The sample sizes are taken in equal increments from $50$ to $200$ when $K=4$, from $50$ to $300$ when $K=6$, from $50$ to $260$ when $K=8$, and from $50$ to $230$ when $K=10$.

\end{itemize}

\end{document}